\documentclass[smallextended, envcountsect]{./springer/svjour3}
\journalname{JOTA}
\smartqed

\usepackage{graphicx}
\usepackage{dependencies/dependencies}
\graphicspath{{images/}}
\begin{document}

\title{Efficient Use of Quantum Linear System Algorithms in Interior Point Methods for Linear Optimization}


\author{Mohammadhossein Mohammadisiahroudi, Ramin Fakhimi, and Tam\'as Terlaky}

\institute{Mohammadhossein Mohammadisiahroudi,  Corresponding author  \at
              Lehigh University \\
              Bethlehem, PA 18015, USA\\
              mom219@Lehigh.edu
            \and
           Ramin Fakhimi \at
             Lehigh University \\
              Bethlehem, PA 18015, USA\\
              fakhimi@Lehigh.edu
           \and
           Tam\'as Terlaky \at
             Lehigh University \\
              Bethlehem, PA 18015, USA\\
              terlaky@Lehigh.edu
}

\date{Received: date / Accepted: date}

\maketitle

\begin{abstract}
Quantum computing has attracted significant interest in the optimization community because it potentially can solve classes of optimization problems faster than conventional supercomputers. 
Several researchers proposed quantum computing methods, especially Quantum Interior Point Methods (QIPMs), to solve convex optimization problems, such as Linear Optimization, Semidefinite Optimization, and Second-order Cone Optimization problems. 
Most of them have applied a Quantum Linear System Algorithm at each iteration to compute a Newton step.
However, using quantum linear solvers in QIPMs comes with many challenges, such as having ill-conditioned systems and the considerable error of quantum solvers.
This paper investigates how one can efficiently use quantum linear solvers in QIPMs. 
Accordingly, an Inexact Infeasible Quantum Interior Point Method is developed to solve linear optimization problems. 
We also discuss how we can get an exact solution by Iterative Refinement without excessive time of quantum solvers. 
Finally, computational results with a QISKIT implementation of our QIPM using quantum simulators are analyzed.

\end{abstract}

\section{Introduction}

Linear Optimization (LO) is defined as optimizing a linear function over a set of linear equality and inequality constraints.
Several algorithms were developed to solve LO problems \citep{dantzig2016linear, Khachiyan1979_polynomial}. 
\textcite{Karmarkar1984_New} developed the foundations of polynomial time and practically efficient Interior Point Methods (IPMs) for solving LO problems.
Since Karmarkar's publication, a large class of theoretically and practically efficient IPMs were developed, see, e.g., \citep{  Terlaky2013_Interior,Roos2005_Interior,Wright1997_Primal}.
Contrary to the Simplex method, a feasible IPM reaches an optimal solution by traversing through the interior of the feasible region \citep{Roos2005_Interior}.

Contemporary IPMs reach an optimal solution by starting from an interior point and following the central path \citep{Roos2005_Interior}.
The most efficient IPMs are primal-dual methods, meaning they strive to satisfy the optimality conditions while maintaining strict primal and dual feasibility. 
It should be noted that basic IPMs need an initial feasible interior point. 
Some current commercial solvers apply Feasible IPMs (F-IPMs) based on the self-dual embedding formulation of LO problems, e.g., MOSEK, while Infeasible Interior Point Methods (I-IPMs) can start with an infeasible but positive solution.
Theoretical analysis shows the best iteration complexity of F-IPMs is $\Ocal(\sqrt{n}L)$, where $n$ is the number of variables, and $L$ is the binary length of the input data.
On the other hand, the best iteration complexity of I-IPMs is $\Ocal(nL)$. 
In practice, the performance of both feasible and infeasible IPM are similar \citep{Wright1997_Primal}. 

A linear equation system is solved at each iteration of IPMs to calculate a Newton direction. 
There are three choices for the linear equation system: (i) Full Newton System,  (ii) Augmented System, and (iii) Normal Equation System (NES).
In classical computers, a general approach is applying Cholesky Factorization to solve the NES, because it has a symmetric positive definite coefficient matrix. 
A partial update technique improved the complexity of solving the system in each iteration. This approach leads to the best total complexity  of $\Ocal(n^3L)$ arithmetic operations for solving LO problems.
However, this is not always efficient in practice \cite{Roos2005_Interior}. 
Since Cholesky Factorization requires $\mathcal{O}(n^3)$ arithmetic operations for large dense matrices, several researchers studied inexact solution methodologies for solving Newton systems. 
\textcite{Bellavia1998_Inexact} proved the convergence of an Inexact I-IPM (II-IPM) for general convex optimization problems. 
Mizuno and his colleagues studied the convergence of II-IPMs \citep{ MizunoJarre1999_Global, Freund1999_Convergence}. 
\textcite{Korzak2000_Convergence} and \textcite{Baryamureeba2006_convergence} proved the convergence of II-IPM proposed by \textcite{Kojima1993_primal}.
\textcite{Korzak2000_Convergence} also showed that the total time complexity of his algorithm is polynomial. 
\textcite{Gondzio2009_Convergence}, and \textcite{ Monteiro2003_Convergence} investigated the use of Preconditioned Conjugate Gradient (PCG) methods in II-IPMs. 
\textcite{Zhou2004_Polynomiality} also proved the convergence of  an II-IPM for SDO problems.
The best iteration complexity of II-IPMs is $\Ocal(n^2L)$ which is $\Theta(n^{1.5})$ weaker than the best iteration complexity of exact  F-IPMs.

Quantum computers have recently emerged as a powerful alternative to classic computers~\citep{Childs2017_Quantum}.
Starting from Deutsch's Problem, a series of problems and algorithms have demonstrated theoretically exponential speedup compared to their classical counterparts \citep{Deutsch1985, Deutsch1992, Simon1997}. 
One of the promising quantum algorithms is the HHL method \citep{Harrow2009_Quantum} to solve linear system problems. 
The HHL method showed quantum advantage with respect to the dimension compared to classical linear systems solutions algorithm. 
However, this method has unfavorable dependence on the condition number, the sparsity of the matrix, and the inverse precision of the solution. 
Several researchers attempted to improve the performance of Quantum Linear System Algorithms (QLSAs), \citep{Wossnig2018_Quantum, Childs2017_Quantum}.
This paper explores an efficient use of the QLSAs for solving the Newton system at each iteration of IPMs. 

To investigate quantum speedup for continuous optimization, Brand{\~a}o and Svore \citep{brandao2017quantumFOCS} proposed a non-IPM quantum algorithm based on the Multiplicative Weight Update Method (MWUM) of \citep{arora2012multiplicative} to solve Semidefinite Optimization (SDO) problems. 
After this papers, many improved versions of Quantum MWUMs (QMWUMs) are proposed for SDO \citep{brandao2017quantum, van2018improvements}, and LO \citep{van2019quantum}. 
The advantage of QMWUMs is that their complexity has linear dependence on the dimension for SDO and sublinear dependence for LO.
The major issue with QMWUMs is that they are highly dependent on inverse precision and an upper bound for the norm of the optimal solution, which are exponentially large for~LO.

In another direction, a few studies were proposing Quantum IPMs (QIPMs).
Kerenidis and his colleagues presented a series of papers on QIPMs for solving LO, SDO \citep{kerenidisParkas2020_quantum}, and Second-order Cone Optimization (SOCO) problems \citep{Kerenidis2021_Quantum}. 
Note that for both SDO and SOCO, they missed some important parts, such as the need for symmetrization/scaling. 
So, their QIPMs are invalid for their conic optimization problems.
\textcite{kerenidisParkas2020_quantum} claimed that their algorithm has  $\Ocal(\frac{n^{2}}{\epsilon^{2}}\kappa^3 \log(\frac{1}{\zeta}))$ complexity for LO problems where $\zeta$ is the final optimality gap, $\epsilon$ is the precision of the QLSA.
Observe that, to reach final precision $\zeta$, the Newton system needs to be solved with precision $\epsilon < \zeta$. 
Further, $\kappa$ is an upper bound for the condition number of the Newton systems at each iteration.
They used block encoding to construct the Newton system and Quantum Tomography Algorithm (QTA) to extract the classical solution.
\textcite{Casares2020_quantum} also provided a hybrid predictor-corrector QIPM scheme to solve LO problems using the well-known predictor-corrector method proposed by \textcite{Ye1994_iteration}. 
In each step, it uses the QLSA proposed by \textcite{Chakraborty2018_power} to solve the Newton systems. 
They claimed that the complexity of their method is $\Ocal(L\sqrt{n}(n+m)\overline{\|M\|}_F \kappa^2 \epsilon^{-2})$, where $m$ is the number of constraints, and $\overline{\|M\|}_F$ is an upper bound to the Frobenius norm of the coefficient matrix in the Newton system at each iteration. 
The authors claimed a quantum speedup with respect to the dimension $n$ compared to classical algorithms.
However, they overlooked some elements in their time complexity. 
Thus, the actual complexity has higher dependence on $n$ and the required precision.

In both mentioned QIPMs, an exact F-IPM framework was used,  regardless of the inherent inexactness of QLSAs. 
So, the proposed time complexities are not attainable, since Inexact IPMs have higher iteration complexity. 
Augustino et al. \citep{augustino2021quantum} addressed this issue and proposed two convergent QIPMs for SDO. 
Their complexity shows polynomial speed-up with respect to dimension, but similar to former QIPMs, their complexity suffers from linear dependence on condition number of the Newton system and precision.
A major reason is that QLSAs' time complexities depend on the condition number of the linear system, and it is linear in inverse precision and dimension to extract an accurate classical solution.
In IPMs, the condition number of the Newton system typically goes to infinity as the algorithm approaches an optimal solution \citep{Roos2005_Interior}.
It is also worth mentioning that $\epsilon$, the required precision for solution of the Newton system, needs to be significantly smaller than the final precision $\zeta$. 
Consequently, these complexity bounds are not polynomial in the classical sense.
This paper explores an efficient use of the QLSAs for solving the Newton system at each iteration of QIPMs. 
We propose an II-QIPM to find an exact solution. 
We also employ an iterative refinement scheme to avoid exponential complexity in finding an exact optimal solution.

This paper is structured as follows.
Section~\ref{sec: Quantum Linear Algebra} discusses the performance of existing QLSAs.
Section~\ref{sec: Problem Definition} introduces the LO problem and its characteristics.
In Section~\ref{sec: II-QIPM}, we present an II-QIPM to solve LO problems, in which a QLSA is used for solving the NES.
We also analyze the complexity of the proposed II-QIPM. 
In Section~\ref{sec: iterative refinement}, we employ an Iterative Refinement scheme to find an exact optimal solution of a LO problem without excessive time of QLSAs.
%
%
Section~\ref{sec: Numerical Experiments} presents the first implementation of the proposed II-QIPM with Iterative Refinement (IR-II-QIPM) and evaluates the algorithms through computational experiments.
Finally, conclusions are presented in Section~\ref{sec: conclusion}.

\section{Quantum Linear Algebra} \label{sec: Quantum Linear Algebra}

This section reviews the use of Quantum algorithms to solve Linear System Problems (LSP). 
\begin{definition}[LSP]\label{def: LSP}
The Linear System Problem: Find a vector $z\in \mathbb{R}^{p}$ such that it satisfies equation $Mz=\sigma$ with coefficient matrix $M\in \mathbb{R}^{p \times p}$ and right-hand side (RHS) vector $\sigma\in \mathbb{R}^p$.
\end{definition}

In the rest of the paper, $\|M\|=\|M\|_2$ is the 2-norm of matrix $M$, and $\|M\|_F$
is the Frobenius norm of $M$.
We sometimes use $\tilde{\Ocal}$ which suppresses the polylogarithmic factors in the "Big-O" notation.
The quantities of the polylogarithmic factors are indicated as subscripts of $\tilde{\Ocal}$.
We use $\Rmbb^n$ for the set of $n$-dimensional vectors of real numbers and $\Cmbb^n$ for the set of $n$-dimensional vectors of complex numbers.
It should be mentioned that the complexity of a classical algorithm when solving LSP means the number of arithmetic operations and the complexity of a quantum algorithm is the number of quantum gates. 

The LSP can have either one, many, or no solutions. 
A basic approach for solving an LSP is Gaussian elimination, or LU factorization,  with $\Ocal(p^3)$ arithmetic operations. 
If $M$ is a square symmetric positive semidefinite (PSD) matrix, we can also apply Cholesky factorization with $\Ocal(p^3)$ arithmetic operations. 
The best complexity for an iterative algorithm with respect to $p$ is $\Ocal(pd\sqrt{\kappa}\log(1/\epsilon))$ arithmetic operations for the Conjugate Gradient method solving systems with symetric PSD matrices, where $d$ is the maximum number of non-zero elements in any row or column of $M$, $\kappa$ is the condition number of $M$, and $\epsilon$ is the error allowed. If matrix $M$ is just symmetric, one can use Lanczos algorithm with higher complexity.
For an LSP with general square matrix $M$, the best iterative method is the GMRES algorithm, which has $\Ocal(n^3)$ worst-case complexity \citep{saad2003iterative}. 
To sum up, LSPs in general form are solvable in polynomial time in classical computing setting, but the worst-case complexity is $\Ocal(n^3)$ for either matrix decomposition methods or iterative methods.

Before discussing QLSAs, we should mention that the $\ket{z}$ notation represent the quantum state corresponding to the unit classical vector $z$. 
We denote the basis state $\ket{i}$, which is a column vector with dimension $p$, one in coordinate $i$ and zero in other coordinates \cite{Childs2017_Quantum}.
QLSAs have different approaches, while all of them are solving Quantum Linear System Problems (QLSPs) defined as follows.
%
\begin{definition}[QLSP]\label{def: QLSP}
Let $M\in \mathbb{C}^{p \times p}$ be a Hermitian matrix with  $\|M\| = 1$, $\sigma \in \mathbb{C}^p$, and $z:= M^{-1} \sigma$. We define quantum states
\begin{equation*}
    \ket{\sigma} = \frac{\sum_{i = 1}^{p}\sigma_i \ket{i}}{\|\sum_{i = 1}^{p}\sigma_i \ket{i}\|}
    \quad \text{ and }\quad
    \ket{z} = \frac{\sum_{i = 1}^{p}z_i \ket{i}}{\|\sum_{i = 1}^{p}z_i \ket{i}\|}.
\end{equation*}
For target precision $0<\epsilon_{QLSP}$, the goal is to find $\ket{\ztilde}$ such that $\|\ket{\ztilde}-\ket{z}\|\leq\epsilon_{QLSP}$, succeeding with probability $\Omega(1)$.
\end{definition}

Based on Definition~\ref{def: LSP} and Definition~\ref{def: QLSP}, QLSP is a different form of the LSP.
It should be noted that at each iteration of QIPMs, instead of an LSP, we need to use a QLSA to solve a QLSP. 
Thus, we need to translate LSP to QLSP, solve the QLSP by QLSA and extract the solution by a QTA. 
Here, we analyze the details and costs of the process of translating LSPs to QLSPs, encoding in the quantum setting, solving them with a QLSA, and extracting classical solutions with a QTA  as follows. 
%
\begin{enumerate}[label=(\roman*)]
\item \underline{Model of Computation}: the first important step is determining how to encode the input data in the quantum setting. There are two major input models.
One is the sparse-access model which is used in the HHL algorithm \cite{Harrow2009_Quantum} and then in other QLSAs \cite{Childs2017_Quantum, vazquez2022enhancing}. 
This is a quantum version of classical sparse matrix computation, and we assume access to unitaries that calculate the index of the $l$th non-zero element of the $k$th row of a matrix $M$ when given $(k,l)$ as input. 
A different input model, now known as the quantum operator input model, is proposed in \textcite{low2019hamiltonian}, which is based on the idea of block-encoded matrices. 
In this input model, one has access to unitaries that store the coefficient matrix:
$$ U = \begin{pmatrix} M/\alpha & \cdot \\
\cdot & \cdot \end{pmatrix},$$
where $\alpha \geq \| M \|$ is a normalization factor chosen to ensure that $U$ has norm at most 1.
\textcite{Chakraborty2018_power} showed that this quantum operator input model is more efficient than the sparse-access model and oracles to encode input data using block-encoding has favorable complexity compared to the sparse-access model.
On the other hand, most of the block-encoding approaches use Quantum Random Access Memory (QRAM). 
However, the sparse-access model can be implemented in the standard gate-based quantum circuit model. 
Despite efficient encoding procedures, the quantum operator input model can not be implemented with current quantum computers since there is no physical implementation of QRAM. 
In our analysis, we assume that the data is stored in QRAM, and we use the quantum operator model by \textcite{Chakraborty2018_power} for the QLSA, which enjoys the best complexity to date.
Using the QRAM structure, \textcite{kerenidis2016quantum} showed that one can implement $\epsilon$-approximate block-encoding of $M$ with $\Ocal(\text{polylog}(\frac{p}{\epsilon}))$ complexity.
Further, given $M$ in the sparse-access
input model, there is an $\epsilon$-approximate block-encoding of $M$ that can be implemented in complexity $\Ocal(\text{polylog}(\frac{p\alpha}{\epsilon}))$ \cite{Chakraborty2018_power}.
Thus, our results also apply to the sparse-access input model if we have the data in that form.
Using these results and assuming access to QRAM, \textcite{Chakraborty2018_power} proposed a QLSA, in which they construct state $\ket{\sigma}$, build and implement a block-encoding of matrix $M$ with $\Ocal(\text{polylog}(\frac{p}{\epsilon}))$ complexity.
Since we are interested in using QLSAs in a hybrid approach, we need to consider the cost of storing data in a classical form to QRAM which is $\Ocal(p^2)$ for a fully dense matrix \cite{Chakraborty2018_power}. 
This cost should be paid once, and in the next section, it is shown that this cost will be dominated by the classical operations in each iteration of QIPMs.
\item \underline{Translating LSP to QLSP}: Based on the definition of QLSP, the coefficient matrix of the system must be Hermitian. 
    If $M$ is not Hermitian, one can construct $\bar{M}\bar{z} =\bar{\sigma}$, where
    \begin{align*}
        \bar{M} = \bbmatrix 0 & M\\ M^\dagger & 0\ebmatrix, \quad
        \bar{\sigma} = \bpmatrix \sigma \\ 0\epmatrix,
    \end{align*}
    and find the vector $\bar{z} = \bpmatrix 0 \\ z\epmatrix$, where $M^\dagger$ denotes the conjugate transpose of $M$. 
    The size of the problem increases from $p$ to $2p$. QLSP assumes $\|M\| = 1$. In the structure of block encoding, we address this normalization, but for sparse encoding, we need to normalize the system $\bar{M}\bar{z} = \bar{\sigma}$ where 
    \begin{align*}
        \bar{M} =\frac{M}{\|M\|},
        \quad 
        \bar{\sigma} = \frac{\sigma}{\|M\|},
        \quad \text{and} \quad
        \bar{z}=z.
    \end{align*}  
    Although we are using block-encoding, we scale the matrix in advance and let $\alpha=1$ in the block encoding. 
    This scaling will affect the precision and complexity, but in this way, we can have better complexity analysis since, in QIPMs, we need to bound the residual.
    We also have a similar scaling in the definitions of states $\ket{z}$ and $\ket{\sigma}$. 
    The corresponding QLSP problem is $\Mbar \ket{z}=\ket{\bar{\sigma}}$, and we can find an inexact solution $\ket{\ztilde}$ with $\|\ket{\ztilde}-\ket{z}\|\leq \epsilon_{QLSP}$. Since $\|\ket{\ztilde}\|=1$, we need to scale back the solution with $\ztilde=\|\bar{\sigma}\|\ket{\ztilde}$.
    This scaling affects the target precision. 
    In IPMs, we are interested in finding a solution where the residual is bounded as $\|\sigma-M\ztilde\|\leq \epsilon_{LSP}$. Thus, we have
    \begin{align*}
        \|\sigma-M\ztilde\|&\leq\|M\|\|\bar{\sigma}-\Mbar \ztilde\|\\
        &\leq \|M\|\|\bar{\sigma}\|\|\ket{\bar{\sigma}}-\Mbar\ket{ \ztilde}\|\\
        &\leq \|M\|\frac{\|\sigma\|}{\|M\|}\|\Mbar \ket{z}-\Mbar\ket{ \ztilde}\|\\
        &\leq \|\sigma\|\|\ket{z}-\ket{ \ztilde}\|\\
        &\leq \|\sigma\|\epsilon_{QLSP}.
    \end{align*}
    Thus, we must set the target error of the QLSP as $\epsilon_{QLSP}=  \epsilon_{LSP}/\|\sigma\|.$
 
    \item \underline{QLSA}: After preprocessing and encoding, we can apply the QLSA to solve the QLSP. The HHL algorithm~\citep{Harrow2009_Quantum} solves QLSPs with $\tilde{\Ocal}_{p}(\frac{d^2\kappa^2}{  \epsilon_{QLSP}})$ complexity. 
Several researchers attempted to improve the performance of the HHL algorithm. 
As the first attempt, Amplitude Amplification decreases the dependence on $\kappa^2$ to $\kappa$ \citep{ambainis2012variable}. 
\textcite{Wossnig2018_Quantum} proposed a QLSA  with $\tilde{\Ocal}_{p}(\|M\|_F\frac{\kappa}{  \epsilon_{QLSP}})$ complexity by using the Quantum Singular Value Estimation.
In another direction, \textcite{Childs2017_Quantum} developed two QLSAs with exponentially better dependence on error with $\tilde{\Ocal}_{p,\kappa,\frac{1}{\epsilon_{QLSP}}}(d\kappa)$ complexity. 
They proposed two approaches using Fourier and Chebyshev series representations.
The best QLSA with respect to complexity uses block encoding and QRAM with $\tilde{\Ocal}_{p,\frac{1}{\epsilon_{QLSP}}}(\alpha\kappa)$ complexity \citep{Chakraborty2018_power}. The normalization factor $\alpha$ is equal to $1$ for our approach since we apply the QLSA to QLSP, which is normalized in advance. 
The details of these methods are out of the scope of this paper. For further details, see \citep{Dervovic2018_Quantum}. 

    \item \underline{QTA}: QLSAs provide a quantum state proportional to the solution. 
    We cannot extract the classical solution by a single measurement. 
    We need Quantum Tomography Algorithms (QTAs) to extract the classical solution. 
    There are several papers improving QTAs, see, e.g., \citep{Kerenidis2021_Quantum}. We used the best QTA by \cite{van2022quantum}, with complexity $\Ocal(\frac{p\varrho}{\epsilon_{QTA}})$, where $\varrho$ is bound for the norm of the solution. In our approach, $\varrho$ is equal to $1$ since we applied QTA to QLSP and scaled back the solution after QTA.  Since the error is additive, we may choose $$\epsilon_{QTA}=\epsilon_{QLSA}=\frac{\epsilon_{LSP}}{2\|\sigma\|}.$$
\end{enumerate}

Table~\ref{tab: LSA-QTA} presents the complexity of different classical and quantum algorithms for solving an LSP. 
Here, the complexity of a QTA  is considered in the complexity of QLSAs.
As shown in Table~\ref{tab: LSA-QTA}, although the complexity of solving LSP using QLSA+QTA will have a similar dependence on $p$ compared to Conjugate Gradient (CG), the complexity of CG method is valid for solving LSPs with PSD matrices and the complexity of QLSA+QTA in Table~\ref{tab: LSA-QTA} are valid for solving LSPs in a general form.
QLSAs have better dependence on dimension $p$ compared to factorization and elimination techniques. Generally, QLSA has worse dependence on $\kappa$, $\frac{1}{\epsilon_{LSP}}$, $\|M\|$, and $\|\sigma\|$.
In our QIPM, we solve a modified NES which has symmetric PSD matrix and both CG and QLSA+QTA are applicable.
Although CG has better complexity than QLSA+QTA based on Table~\ref{tab: LSA-QTA}, we use QLSA+QTA to solve the NES since it enables us to also build the Newton system in quantum computer which has polylogaritimic dependence on $n$. 
This approach enable us to avoid matrix-matrix product, which is inevitable in classical IPM using CG. 
Thus, the cost per iteration of IPM can leverage the quadratic dependence on dimension but the linear dependence on condition number and inverse precision can lead to exponential complexity.
The following sections will discuss how we can deal with errors and condition numbers when we use QLSA in IPMs.
\begin{table}[ht]
\centering
\caption{Complexity of solving an LSP.}\label{tab: LSA-QTA}
\begin{tabular}{|l|c|l|} \hline 
 \multicolumn{1}{|c|}{Algorithm}  &  Complexity & Comment \\ \hline
Factorization methods (e.g. LU) &$\Ocal\big(p^{3}\big)$
&\\  
Conjugate Gradient &$\Ocal\big(pd\sqrt{\kappa}\log(\frac{1}{\epsilon})\big)$
& For LSPs with PSD matrices.\\  
HHL \citep{Harrow2009_Quantum}  + QTA \citep{van2022quantum}                 &$\Ocal\Big(pd+p\frac{d^2\kappa^2\|\sigma\|^2}{\|M\|^2\epsilon^2}\text{polylog}(\frac{p\kappa}{\epsilon}) \Big)$
& Using sparse encoding.\\ 
QLSA \citep{Childs2017_Quantum} + QTA \citep{van2022quantum}          &  $\Ocal\Big(pd+p \frac{d\kappa\|\sigma\|}{\|M\|\epsilon}\text{polylog}(\frac{p\kappa}{\epsilon}) \Big)$
& Using sparse encoding.\\ 
QLSA \citep{Chakraborty2018_power} + QTA \citep{van2022quantum}         & $\Ocal\Big(p^2+p \frac{\kappa\|\sigma\|}{\epsilon}\text{polylog}(\frac{p}{\epsilon})\Big)$
& Using block encoding.\\ \hline
\end{tabular}
\end{table}
%
\section{The Linear Optimization Problem} \label{sec: Problem Definition}
%
Here we consider the standard form of Linear Optimization (LO) problem as follows:
\begin{equation} \label{eq:primal and dual problem}
    \text{(P)}\quad
    \begin{aligned}
    \min\  c^T&x \\
    \text{s.t. }
    Ax &= b, \\
    x &\geq 0,
    \end{aligned}
    \qquad \qquad \text{(D)} \quad
    \begin{aligned}
    \max \  b^Ty\ \ & \\
    \text{s.t. }
    A^Ty +&s = c,\\
    &s \geq 0,
    \end{aligned}
\end{equation}
where $A:m\times n$ matrix with rank$(A)=m$, vectors $y,b\in \mathbb{R}^m$, and $x,s,c\in \mathbb{R}^n$.
Problem (P) is called the primal problem and (D) is called the dual problem.
Due to the Strong Duality Theorem \cite{Roos2005_Interior}, all optimal solutions, if exist, belong to the set $\mathcal{PD}^*$, which is defined as
\begin{equation*}
\mathcal{PD}^*=\left\{(x,y,s)\in\mathbb{R}^{n+m+n}:\ Ax=b,\ A^Ty+s=c,\ x^Ts=0, \ (x,s)\geq0\right\}.
\end{equation*}
Now, we can define the optimal partition of the LO problem as 
\begin{align*}
    B&=\{j\in\{1,\dots,n\}:x^*_j>0 \text{ for some }(x^*,y^*, s^*)\in \mathcal{PD}^*\},\\
    N&=\{j\in\{1,\dots,n\}:s^*_j>0 \text{ for some }(x^*,y^*, s^*)\in \mathcal{PD}^*\}.
\end{align*}
It is shown that $B\cup N=\{1,\dots,n\}$, and $B\cap N=\varnothing$ \cite{Roos2005_Interior}.
\begin{assumption}
From now on, without loss of generality \citep{Roos2005_Interior}, we assume that the Interior Point Condition (IPC) holds, i.e., there exists a solution $(x,y,s)$ such that   
\begin{equation*}
    Ax=b,\ x>0,\ A^Ty+s=c\text{, and }s>0.
\end{equation*}
\end{assumption}
The IPC warranties \cite{Roos2005_Interior} that the optimal set $\mathcal{PD}^*$ is bounded, so there exists $\omega\geq 1$ such that 
\begin{equation*}
    \omega\geq \max\{\|(x^*,s^*)\|_{\infty}:(x^*,s^*)\in \mathcal{PD}^*\}.
\end{equation*}
The central path is the curve defined by
\begin{equation*}
\mathcal{CP}=\left\{(x,y,s)\in\mathbb{R}^{n+m+n}:\ Ax=b,\ A^Ty+s=c,\ x_is_i=\mu \text{ for }i \in\{1,\dots,n\}, \ x,s,\mu>0\right\}.
\end{equation*}
By the IPC, the central path is well-defined, and an interior feasible solution $(x,y,s)$ exists for all $\mu>0$. 
Now, for any $0<\gamma_1<1$ and $1\leq\gamma_2$, we define an infeasible neighborhood of the central path for $\mu=\frac{x^Ts}{n}$ as the following definition in \citep{Wright1997_Primal}. 
\begin{align*}
    \mathcal{N}(\gamma_1,\gamma_2)=\{(x,y,s)\in\mathbb{R}^{n+m+n}:\ &(x,s)>0,\ x_is_i\geq\gamma_1\mu \  \text{ for }\ i\in\{1,\dots,n\},\|(R_P,R_D)\|\leq\gamma_2\mu\},
\end{align*}
where $R_P=b-Ax$, and $R_D=c-A^Ty-s$.
Assuming that the input data is integral, we denote the binary length of the input data by 
$$L=mn+m+n+\sum_{i,j}\lceil\log(|a_{ij}|+1)\rceil+\sum_{i}\lceil\log(|c_{i}|+1)\rceil+\sum_{j}\lceil\log(|b_{j}|+1)\rceil.$$
The following lemma is a classical result first proved by \textcite{Khachiyan1979_polynomial}.
\begin{lemma}\label{lemma: L bound}
Let $(x^*,y^*,s^*)\in \mathcal{PD}^*$ be a basic solution. If $x_i^*>0$, then we have $x_i^*\geq 2^{-L}$. If $s_i^*>0$, then we have $s_i^*\geq 2^{-L}$.

\end{lemma}
\begin{theorem}[Chapter 3 in \cite{Wright1997_Primal}]\label{cor: exact sol}
An exact optimal solution can be obtained by a strongly polynomial rounding procedure when $(x,y,s)\in \mathcal{N}(\gamma_1,\gamma_2)$ and $\mu\leq2^{-2L}.$
\end{theorem}

If the IPC holds, then the optimal set $\mathcal{PD}^*$ is bounded, and we can find the upper bound for all the coordinates of all optimal solutions as described in the following lemma.
\begin{lemma}[Chapter 5 in \cite{ye2011interior}]\label{lemma: upperbound}
Assuming the IPC, then for any $(x^*,y^*,s^*)\in \mathcal{PD}^*$, 
$
    \max_i\{x_i^*\}\leq 2^{L},\text{ and }
    \max_i\{s_i^*\}\leq 2^{L}.
$
\end{lemma}
For the theoretical purpose, we can use $\omega=2^{L}$ from Lemma~\ref{lemma: upperbound}, but in practice, for concrete LO problems, we may find a smaller bound.
We define the set of $\zeta$-optimal solutions as
\begin{equation*}
    \Pcal\Dcal(\zeta)=\{(x,y,s)\in\mathbb{R}^{n+m+n}:\ (x,s)\geq 0,\ \frac{x^Ts}{n}\leq\zeta ,\|(R_P,R_D)\|\leq\zeta\}
\end{equation*}

%
\section{An Inexact Infeasible Quantum IPM} \label{sec: II-QIPM}
%
To speed up IPMs, we use QLSAs to solve the Newton system at each iteration of IPMs. 
As discussed in Section~\ref{sec: Quantum Linear Algebra}, QLSAs inherently produce inexact solutions. 
Thus, one approach to use QLSA efficiently is to develop an Inexact Infeasible QIPM (II-QIPM).
In this paper, we utilize the KMM method proposed by \textcite{Kojima1993_primal} with the inexact Newton steps calculated by a QLSA.
Given $(x^k,y^k,s^k)\in \mathcal{N}(\gamma_1,\gamma_2)$, let $\mu^k=\frac{(x^k)^Ts^k}{n}$ and $0<\beta_1<1$ be the centering parameter, then the Newton system is defined as
\begin{equation}\label{eq:newton system}
\begin{aligned}
    A\Delta x^k&=b-Ax^k,\\
    A^T \Delta y^k +\Delta s^k&=c-A^Ty^k-s^k,\\
    X^k\Delta s^k+S^k\Delta x^k&=\beta_1 \mu^k e-X^ks^k,
\end{aligned} 
\end{equation}
where $e$ is all one vector with appropriate dimension, $X^k=\text{diag}(x^k)$, and $S^k=\text{diag}(s^k)$.
Instead of solving the full Newton system, we may solve the Augmented system or the Normal Equation System (NES).
%
%
From the Newton system~\eqref{eq:newton system}, the NES is formulated as
\begin{equation}\tag{NES}\label{eq:normal equation}
    M^k\Delta y^k= \sigma^k,
\end{equation}
where 
\begin{align*}
    D^k&=(X^k)^{1/2}(S^{k})^{-1/2},\\
    M^k&=A(D^k)^2A^T,\\
\sigma^k&=A(D^k)^2c-A(D^k)^2A^Ty^k-\beta_1\mu^k A(S^{k})^{-1}e +b-Ax^k\\
&=b-\beta_1 \mu^k A(S^k)^{-1}e+A(D^k)^2(c-A^Ty^k-s^k).
\end{align*}
As we can see, the NES has a smaller size, $m$, than the full Newton system.
Further, the coefficient matrix of the NES is symmetric and positive definite, thus Hermitian. Consequently,  QLSAs can solve the NES efficiently.
By its nature, a QLSA generates an inexact solution $\widetilde{\Delta y}{}^k$ with error bound $\|\Delta y^k-\widetilde{\Delta y}{}^k\|\leq \epsilon^k$. 
This error leads to residual $r^k$ as
\begin{equation*}M^k\widetilde{\Delta y}{}^k=\sigma^k+r^k,\end{equation*}
where $r^k=M^k(\widetilde{\Delta y}{}^k-\Delta y^k)$.
After finding $\widetilde{\Delta y}{}^k$ inexactly by solving the NES using QLSA, we compute the inexact $\widetilde{\Delta x}{}^k$ and $\widetilde{\Delta s}{}^k$ classically as
\begin{equation}
\begin{aligned}
\label{eq:normal equation solutions}
\widetilde{\Delta s}{}^k&=c-A^Ty^k-s^k-A^T\widetilde{\Delta y}{}^k,\\
\widetilde{\Delta x}{}^k&=\beta_1\mu^k (S^{k})^{-1}e-x^k-(D^k)^2\widetilde{\Delta s}{}^k.
\end{aligned}
\end{equation}
As $\widetilde{\Delta s}{}^k$ and $\widetilde{\Delta x}{}^k$ are directly calculated by equations \eqref{eq:normal equation solutions}, one can verify that $(\widetilde{\Delta x}{}^k, \widetilde{\Delta s}{}^k, \widetilde{\Delta y}{}^k)$ satisfies 
\begin{equation} \label{eq:approximate newton system}
    \begin{aligned}
    A\widetilde{\Delta x}{}^k&=b-Ax^k+r^k,\\
    A^T \widetilde{\Delta y}{}^k +\widetilde{\Delta s}{}^k&=c-A^Ty^k-s^k,\\
    X^k\widetilde{\Delta s}{}^k+S^k\widetilde{\Delta x}{}^k&=\beta_1 \mu^k e-X^ks^k.
\end{aligned}
\end{equation}

To have an II-IPM using \eqref{eq:normal equation} with iteration complexity $\Ocal(n^2L)$, the residual norm must decrease at least $\Ocal(\lambda_{\text{min}}(A)
\sqrt{n} \log n)$ time faster than $(x^k )^T s^k$ where $\lambda_{\text{min}}(A)$ is the smallest singular value of $A$ \cite{Zhou2004_Polynomiality}. We can have wider residual bound but with higher iteration complexity of II-IPM \cite{bellavia2004convergence}. In the literature of preconditioning the NES, some papers modified the equations \eqref{eq:normal equation solutions} and the \eqref{eq:normal equation}  to transfer the residual from the first equation of \eqref{eq:approximate newton system} to its last equation. By these changes, we can get much better bounds  \citep{Gondzio2009_Convergence,Monteiro2003_Convergence}. Since tight residual bound leads to the high complexity of QLSA+QTA, in this paper, we use a modification of the NES, which leads to $\Ocal(n^2L)$ iteration complexity of II-QIPM, where the residual is decreasing with the rate of $\Ocal(\sqrt{\mu^k})$.

Since $A$ has full row rank, one can choose an arbitrary basis $\Bhat$, and calculate $A_{\Bhat}^{-1}$, $\Ahat=A_{\Bhat}^{-1}A$, and $\bhat=A_{\Bhat}^{-1}b$. This calculation needs $\Ocal(m^2n)$ arithmetic operations and happens just one time before the iterations of IPM. The cost of this preprocessing is dominated by the cost of II-QIPM, but it can be reduced by using the structure of $A$. For example, if the problem is in the canonical form, there is no need for this preprocessing. In the rest of this paper, all methodology is applied to the preprocessed problem with input data $(\hat{A},\hat{b},c)$.
Now, we can modify the \eqref{eq:normal equation} to 
\begin{equation*}\tag{MNES}\label{eq: modified normal equation}
    \Mhat^k z^k=\hat{\sigma}^k
\end{equation*}
where
\begin{align*}
    \Mhat^k&=(D_{\Bhat}^{k})^{-1}A_{\Bhat}^{-1}M^k((D_{\Bhat}^{k})^{-1}A_{\Bhat}^{-1})^T=(D_{\Bhat}^{k})^{-1}\Ahat (D^k)^2((D_{\Bhat}^{k})^{-1}\Ahat)^T,\\
\hat{\sigma}^k&=(D_{\Bhat}^{k})^{-1}A_{\Bhat}^{-1}\sigma^k=(D_{\Bhat}^{k})^{-1}\bhat-\beta_1 \mu^k (D_{\Bhat}^{k})^{-1}\Ahat(S^k)^{-1}e+(D_{\Bhat}^{k})^{-1}\Ahat(D^k)^2(c-A^Ty^k-s^k),
\end{align*}
where $D_{\Bhat}^k=(X_{\Bhat}^k)^{1/2}(S_{\Bhat}^{k})^{-1/2}$. We use the following procedure to find the Newton direction by solving \eqref{eq: modified normal equation} inexactly with QLSA+QTA.
\begin{enumerate}
    \item []\textbf{Step 1.} Find $\tilde{z}^k$ such that $\Mhat^k\tilde{z}^k=\hat{\sigma}^k+\rhat^k$ and $\|\rhat^k\|\leq\eta\frac{\sqrt{\mu^k}}{\sqrt{n}}$.
    \item []\textbf{Step 2.} Calculate $\widetilde{\Delta y}{}^k=((D_{\Bhat}^{k})^{-1}A_{\Bhat}^{-1})^T\tilde{z}^k$.
    \item []\textbf{Step 3.} Calculate $v^k=(v^k_{\Bhat},v^k_{\Nhat})=(D_{\Bhat}^{k}\rhat^k,0)$.
    \item []\textbf{Step 4.} Calculate $\widetilde{\Delta s}{}^k=c-A^Ty^k-s^k-A^T\widetilde{\Delta y}{}^k$.
    \item []\textbf{Step 5.} Calculate $\widetilde{\Delta x}{}^k=\beta_1\mu^k (S^{k})^{-1}e-x^k-(D^k)^2\widetilde{\Delta s}{}^k-v^k$.
\end{enumerate}
The following Lemma shows how the inexact solution of \eqref{eq: modified normal equation} leads to residual only in the last equation of the Newton system.
\begin{lemma}
For the Newton direction $(\widetilde{\Delta x}{}^k,\widetilde{\Delta y}{}^k,\widetilde{\Delta s}{}^k)$, we have
\begin{equation} \label{eq:approximate modified newton system}
    \begin{aligned}
    A\widetilde{\Delta x}{}^k&=b-Ax^k,\\
    A^T \widetilde{\Delta y}{}^k +\widetilde{\Delta s}{}^k&=c-A^Ty^k-s^k,\\
    X^k\widetilde{\Delta s}{}^k+S^k\widetilde{\Delta x}{}^k&=\beta_1 \mu^k e-X^ks^k-S^kv^k.
\end{aligned}
\end{equation}
\end{lemma}
\begin{proof}
For the Newton direction $(\widetilde{\Delta x}{}^k,\widetilde{\Delta y}{}^k,\widetilde{\Delta s}{}^k)$, one can verify that
\begin{align*}
    \Mhat^k\tilde{z}^k&=\hat{\sigma}^k+\rhat^k\\
    M^k\widetilde{\Delta y}{}^k&=\sigma^k+A_{\Bhat}D_{\Bhat}^{k}\rhat^k.
\end{align*}
For the first equation of \eqref{eq:approximate modified newton system}, we can write
\begin{align*}
    A\widetilde{\Delta x}{}^k=&A(\beta_1\mu^k (S^{k})^{-1}e-x^k-(S^{k})^{-1}X^k\widetilde{\Delta s}{}^k-v^k)\\
    =&A(\beta_1\mu^k (S^{k})^{-1}e-x^k-(S^{k})^{-1}X^k(c-s^k-A^Ty^k-A^T\widetilde{\Delta y}{}^k)-v^k)\\
    =&\beta_1\mu^k A(S^{k})^{-1}e-Ax^k-A(S^{k})^{-1}X^kc+A(S^{k})^{-1}X^ks^k+A(S^{k})^{-1}X^k A^Ty^k\\
    &+A(S^{k})^{-1}X^k A^T\widetilde{\Delta y}{}^k-Av^k\\
    =&\beta_1\mu^k A(S^{k})^{-1}e-A(S^{k})^{-1}X^kc+A(S^{k})^{-1}X^k A^Ty^k+\sigma^k+A_{\Bhat}D_{\Bhat}^{k}\rhat^k-A_{\Bhat}D_{\Bhat}^{k}\rhat^k\\
    =&b-Ax^k .
\end{align*}
The second and third equations of \eqref{eq:approximate modified newton system} are obtained by Steps 4 and 5.
\end{proof}
To have a convergent IPM, we need $\|S^kv^k\|_{\infty}\leq \eta \mu^k$, where $0\leq\eta< 1$ is an enforcing parameter. 
\begin{lemma}
For the Newton direction $(\widetilde{\Delta x}{}^k,\widetilde{\Delta y}{}^k,\widetilde{\Delta s}{}^k)$, if the residual $\|\rhat^k\|\leq\eta\frac{\sqrt{\mu^k}}{\sqrt{n}}$, then $\|S^kv^k\|_{\infty}\leq \eta \mu^k$.
\end{lemma}
\begin{proof}
We have
\begin{align*}
    \|S^kv^k\|_{\infty}=\|S^k_{\Bhat}v^k_{\Bhat}\|_{\infty}=\|S^k_{\Bhat}D_{\Bhat}^{k}\rhat^k\|_{\infty}= \|(S^k_{\Bhat})^{1/2}(X_{\Bhat}^{k})^{1/2}\|_{\infty}\|\rhat^k\|_{\infty}\leq\sqrt{n\mu^k}\|\rhat^k\|\leq\eta\mu^k.
\end{align*}
\end{proof}
In the following, we show that by satisfying $\|\rhat^k\|\leq\eta\frac{\sqrt{\mu^k}}{\sqrt{n}}$, then the iterations of the II-QIPM remain in the $\Ncal(\gamma_1,\gamma_2)$ neighborhood of the central path. 
The following theorem presents the complexity of solving the \eqref{eq: modified normal equation} system  by utilizing the QLSA  of \textcite{Chakraborty2018_power}. 
We can also use other QLSAs discussed in Section~\ref{sec: Quantum Linear Algebra}, leading to different complexity bounds.
\begin{lemma} \label{theo: QLSA solution}
 The QLSA by \cite{Chakraborty2018_power} and the QTA by \cite{van2022quantum} can build the \eqref{eq: modified normal equation} system, and produce a solution ${\widetilde{z}}^k$ for the \eqref{eq: modified normal equation} system satisfying $\|\rhat^k\|\leq\eta\frac{\sqrt{\mu^k}}{\sqrt{n}}$ with $\tilde{\Ocal}_{n,\frac{1}{\mu^k},\|\hat{\sigma}^k\|}(mn+m\sqrt{n}\frac{\kappa_E^k\|\hat{\sigma}^k\|}{\sqrt{\mu^k}})$ complexity, where $E^k=(D^k_{\Bhat})^{-1}\Ahat D^k$, and $\kappa_E^k$ is the condition number of $E^k$.
\end{lemma}
\begin{proof}
Building the \eqref{eq: modified normal equation} system in classical computer needs some matrix multiplications, which costs $\Ocal(m^2n)$ arithmetic operations.
We can write \eqref{eq: modified normal equation} as $E^k(E^k)^T \tilde{z}{}^k=\hat{\sigma}^k$. As we can see, calculating $E^k$ and $\hat{\sigma}^k$ needs just $\Ocal(mn)$ arithmetic operations. \textcite{Chakraborty2018_power} proposed an efficient way to build and solve a linear system in the form $E^k(E^k)^T \tilde{z}{}^k=\hat{\sigma}^k$, with $\tilde{\Ocal}(\text{polylog}(\frac{n}{\epsilon_{QLSA}^k})\kappa_E^k)$ complexity.
Also, we need to find the target precision for QLSA and QTA such that  $\|\rhat^k\|\leq\eta\frac{\sqrt{\mu^k}}{\sqrt{n}}$ is satisfied. 
Thus,  to have $\|\rhat^k\|\leq\eta\frac{\sqrt{\mu^k}}{\sqrt{n}}$, it is sufficient to require $\epsilon_{LSP}^k\leq\eta \frac{\sqrt{\mu^k}}{\sqrt{n}\|\Mhat^k\|}$.
Based on the discussion of Section \ref{sec: Quantum Linear Algebra}, we need to have
\begin{align*}
    \epsilon_{QLSP}^k\leq\eta \frac{\sqrt{\mu^k}}{\sqrt{n}\|\hat{\sigma}^k\|},\quad
    \epsilon_{QLSA}^k\leq\eta \frac{\sqrt{\mu^k}}{2\sqrt{n}\|\hat{\sigma}^k\|},\text{ and }
    \epsilon_{QTA}^k\leq\eta \frac{\sqrt{\mu^k}}{2\sqrt{n}\|\hat{\sigma}^k\|}.
\end{align*}
With this target precision, the QLSA by \cite{Chakraborty2018_power} has $\Ocal(\text{polylog}(\frac{n^{1.5}\|\hat{\sigma}^k\|}{\sqrt{\mu^k}})\kappa_E^k)$ complexity, and the QTA by \citep{van2022quantum} has $\Ocal(\frac{m\sqrt{n}\|\hat{\sigma}^k\|}{\sqrt{\mu^k}})$ complexity. Since calculating $E^k$ and $\hat{\sigma}^k$ classicaly needs $\Ocal(mn)$ arithmetic operations and the cost of storing them in QRAM is also $\Ocal(mn)$, the total cost of building and solving  the \eqref{eq:normal equation} system is $$\tilde{\Ocal}_{n,\frac{1}{\mu^k},\|\hat{\sigma}^k\|}(mn+m\sqrt{n}\frac{\kappa_E^k\|\hat{\sigma}^k\|}{\sqrt{\mu^k}}).$$
The proof is complete.
\end{proof}
We present the II-QIPM as Algorithm~\ref{alg: II-QIPM} for solving LO problems. In this algorithm, we use QLSA and QTA to solve the NES. 
\begin{algorithm}[ht] 
\caption{II-QIPM } \label{alg: II-QIPM}
\begin{algorithmic}[1]
 \State Choose $\zeta >0$, $\gamma_1\in(0,1)$,$\gamma_2>0$,$0<\eta<\beta_1<\beta_2<1$, 
 \State Choose $\omega \geq\max\{1,\|x^*,s^*\|_{\infty}\}$.
 \State $k \gets 0$, $(x^0, y^0, s^0) \gets (\omega e, 0 e, \omega e) $, and $\gamma_2 \gets \max\left\{1, \frac{\|(R_p^0,R_D^0)\|}{\mu^0}\right\}$  
\While {$(x^k,y^k,s^k)\notin \mathcal{PD}_\zeta$}
\State $\mu^k \gets \frac{(x^k)^Ts^k}{n}$\label{alg-step:mu^k}
\State $E^k\gets (D_{\Bhat}^{k})^{-1}\Ahat D^k$ and $\hat{\sigma}^k\gets (D_{\Bhat}^{k})^{-1}\bhat-\beta_1 \mu^k (D_{\Bhat}^{k})^{-1}\Ahat(S^k)^{-1}e+(D_{\Bhat}^{k})^{-1}\Ahat(D^k)^2(c-A^Ty^k-s^k)$ \label{alg-step:NES build}
\State $\epsilon_{QLSA}^k \gets \eta \frac{\sqrt{\mu^k}}{2\sqrt{n}\|\sigma^k\|}$ and 
 $\epsilon_{QTA}^k \gets \eta \frac{\sqrt{\mu^k}}{2\sqrt{n}\|\sigma^k\|}$ \label{alg:QTA error}
\State $(\Delta x^k,\Delta y^k,\Delta s^k) \gets$ \textbf{solve} \ref{eq: modified normal equation}($\beta_1$) by QLSA+QTA with precision $\epsilon_{QLSA}^k \text{ and }\epsilon_{QTA}^k$ \label{alg-step:NES solve}
\State $\begin{aligned}
    \hat{\alpha}^k\gets \max \Big\{\bar{\alpha} \in [0,1] \ | \ & \text{for all } \alpha \in [0,\bar{\alpha}]\text{ we have }\\ &\big((x^k,y^k,s^k)+\alpha(\Delta x^k,\Delta y^k,\Delta s^k)\big) \in \mathcal{N}(\gamma_1, \gamma_2) \text{ and } \\
    &(x^k+\alpha \Delta x^k)^T(s^k+\alpha \Delta s^k)\leq \big(1-\alpha (1-\beta_2)\big)(x^k)^Ts^k \Big\}
\end{aligned}$ \label{alg-step: alpha_hat^k}
\State $(x^{k+1},y^{k+1},s^{k+1}) \gets (x^k,y^k,s^k)+\hat{\alpha}^k(\Delta x^k,\Delta y^k,\Delta s^k)$\label{alg-step:update solution}
\If {$\|x_{k+1},s_{k+1}\|_{\infty}>\omega$} \label{alg-step:check infeasibility}
\State \Return {Primal or dual is infeasible.}
\EndIf
\State $k \gets k+1$
\EndWhile
\State \Return{$(x^k,y^k,s^k)$}
\end{algorithmic}
\end{algorithm}
%
%

%
It can be easily verified that e.g., $\beta_1=0.5$, $\beta_2=0.9995$, $\eta=0.4$ and $\gamma_1=0.5$ yield a valid choice, i.e. satisfying the conditions in the first line of Algorithm~\ref{alg: II-QIPM}.
In the following, we prove the polynomial complexity of Algorithm~\ref{alg: II-QIPM}.

%
\subsection{Convergence of the II-QIPM} \label{sec: Convergence of II-QIPM}
%
In this section, a convegence proof and iteration complexity bound are provided for the proposed II-QIPM. 
The analysis closest to ours can be found in \citep{Gondzio2009_Convergence} and \citep{Monteiro2003_Convergence} where they are preconditioning the Newton systems which are then solved by CG.
However, we modified the NES and solved it by QLSA+QTA under different conditions and assumptions.
First, in Lemma~\ref{lemma: iterations properties}, we study basic properties of the proposed II-QIPM as presented in Algorithm~\ref{alg: II-QIPM}.
Then, Lemma~\ref{lemma: bound of alpha} shows that the sequence $\{\hat{\alpha}^k\}$ is strictly positive for all $k$. 
The iteration complexity of Algorithm~\ref{alg: II-QIPM} is proved in Theorem~\ref{theo: convergence}.
In the proof, we need to analyze values of $\alpha^k$ such that at iteration $k$ of Algorithm \ref{alg: II-QIPM}, the Newton step satisfies all the conditions in line \ref{alg-step: alpha_hat^k}. 
To ease notation, we use
\begin{align*}
    x^k(\alpha)=x^k+\alpha \Delta x^k,\   y^k(\alpha)=y^k+\alpha\Delta y^k,&\   s^k(\alpha)=s^k+\alpha\Delta s^k,\  
    \mu^k(\alpha)=\frac{x^k(\alpha)^Ts^k(\alpha)}{n},\\
    R_P^k(\alpha)=b-Ax^k(\alpha), \text{ and }&R_D^k(\alpha)=c-A^Ty^k(\alpha)-s^k(\alpha).
\end{align*}
The following lemma shows some properties of the proposed II-QIPM.
\begin{lemma}\label{lemma: iterations properties}
At iteration $k$ of Algorithm \ref{alg: II-QIPM}, for any $\alpha \in [0,1]$, with $\mu^k=\frac{(x^k)^Ts^k}{n}$, we have

\noindent \resizebox{1\linewidth}{!}{\begin{minipage}{\linewidth} 
\begin{subequations} 
	\begin{alignat}{2}
    R_P^k(\alpha)&= (1-\alpha )R_P^k,\label{eq:primal relationship}\\
    R_D^k(\alpha)&= (1-\alpha)R_D^k,\label{eq:dual relationship}\\
    (x^k(\alpha))^Ts^k(\alpha) &\geq (1+\alpha(\beta_1-\eta-1))n\mu^k+\alpha^2(\Delta  x^k)^T\Delta s^k,\label{eq:complementarity relationship}\\
    x_i^k(\alpha)s_i^k(\alpha) &\geq(1-\alpha)x^k_is^k_i+\alpha (\beta_1 -\eta)\mu^k + \alpha^2 \Delta x^k_i\Delta s^k_i \text{ for } i \in \{1,2,\dots,n\}. \label{eq:i-th complementarity relationship} 
\end{alignat}
\end{subequations}
\end{minipage}}

\end{lemma}
\begin{proof}

To prove \eqref{eq:primal relationship} and \eqref{eq:dual relationship}, for any $\alpha \in [0,1]$, by~\eqref{eq:newton system} we have
\begin{align*}
    R_P^k(\alpha) &= b-A(x^k+\alpha \Delta x^k) =b-Ax^k-\alpha A\Delta x^k=b-Ax^k-\alpha(b-Ax^k) = (1-\alpha )R_P^k,\\
    R_D^k(\alpha) &= c-A^T(y^k+\alpha \Delta y^k)-s^k-\alpha \Delta s^k=c-A^Ty^{k}-s^{k}-\alpha(A^T \Delta y+\Delta s)= (1-\alpha)R_D^k.
\end{align*}
To prove \eqref{eq:complementarity relationship}, using \eqref{eq:approximate modified newton system}, we have
\begin{equation*}
\begin{aligned}
    (x^k+\alpha \Delta x^k)^T(s^k+\alpha \Delta s^k)
    &=(x^k)^Ts^k+\alpha [(x^k)^T\Delta s^k+(s^{k})^T\Delta x^k]+ \alpha^2(\Delta  x^k)^T\Delta s^k\\
    &\geq (x^k)^Ts^k+\alpha [n\beta_1\mu^k-(x^k)^Ts^k-n\eta\mu^k]+ \alpha^2(\Delta  x^k)^T\Delta s^k\\
    &= [1+\alpha(\beta_1-\eta-1)](x^k)^Ts^k+\alpha^2(\Delta  x^k)^T\Delta s^k.
\end{aligned}
\end{equation*}
Using \eqref{eq:approximate modified newton system} again, we can similarly prove \eqref{eq:i-th complementarity relationship} for all $i\in\{1,2,\dots,n\}$ as follows:
\begin{equation*}\label{eq:complementarity lemma3 last part}
\begin{aligned}
    (x^k_i+\alpha \Delta x^k_i)(s^k+\alpha \Delta s^k_i)
    &=x^k_is^k_i+\alpha (x^k_i\Delta s^k_i+s^k_i\Delta x^k_i)+ \alpha^2\Delta x^k_i\Delta s^k_i\\
    &\geq x^k_is^k_i+\alpha (\beta_1\mu^k-x^k_is^k_i-\eta\mu^k)+ \alpha^2\Delta x^k_i\Delta s^k_i\\
    &= (1-\alpha)x^k_is^k_i+\alpha (\beta_1 -\eta)\mu^k+\alpha^2 \Delta x^k_i\Delta s^k_i.
\end{aligned}
\end{equation*}
Thus, the proof is complete.
\end{proof} 
Let us define the following functions
\begin{align*}
    G^k_i(\alpha)&=x^k_i(\alpha)s^k_i(\alpha)- \gamma_1\mu^k(\alpha)\text{ for }i\in \{1,\dots,n\},\\ 
    g^k(\alpha)&=x^k(\alpha)^Ts^k(\alpha)- (1-\alpha)(x^k)^Ts^k,\\ 
    h^{k}(\alpha)&=\big(1-\alpha(1-\beta_2)\big)(x^k)^Ts^k-x^k(\alpha)^Ts^k(\alpha).
\end{align*}
We use the defined functions to check if a step with length $\alpha$ sufficiently reduces the complementarity gap and keeps the next iterate in the neighborhood $\Ncal(\gamma_1,\gamma_2)$.
It is obvious that $h^k(\alpha)\geq0$ means the Armijo condition
$(x^k(\alpha))^Ts^k(\alpha)\leq \big(1-\alpha (1-\beta_2)\big)(x^k)^Ts^k$ holds, and the next lemma shows how we check that an iterate is in the neighborhood of the central path.
\begin{lemma}\label{lem: neigh}
For step length $0<\alpha\leq1$, if $G^k_i(\alpha)\geq0$ and $g^k(\alpha)\geq0$ then $\big(x^k(\alpha),y^k(\alpha),s^k(\alpha)\big)\in\Ncal(\gamma_1,\gamma_2)$.
\end{lemma}
\begin{proof}
It is easy to verify that that conditions $G^k_i(\alpha)\geq0$ and $g^k(\alpha)\geq0$ lead to 
\begin{align*}
x^k_i(\alpha)s^k_i(\alpha)&\geq \gamma_1\mu^k(\alpha)\text{ for }i\in \{1,\dots,n\},\\
    (x^k+\alpha \Delta x^k)^T(s^k+\alpha \Delta s^k)&\leq \big(1-\alpha (1-\beta_2)\big)(x^k)^Ts^k,
\end{align*}
respectively. 
Since  $g^k(\alpha)\geq0$, i.e. $x^k(\alpha)^Ts^k(\alpha)\geq (1-\alpha)(x^k)^Ts^k$, we have
$$\frac{\big\|\big(R_P^k(\alpha), R_D^k(\alpha)\big)\big\|}{\|(R_P^0,R_D^0)\|}
=\frac{(1-\alpha)\|(R_P^k,R_D^k)\|}{\|(R_P^0,R_D^0)\|}
\leq\frac{(1-\alpha)\mu^k}{\mu^0}
\leq\frac{\mu^k(\alpha)}{\mu^0}.$$
Further, as $\gamma_2=\frac{\|(R_P^0,R_D^0)\|}{\mu^0}$, we can conclude that $\big(x^k(\alpha),y^k(\alpha),s^k(\alpha)\big)\in\Ncal(\gamma_1,\gamma_2)$.
\end{proof}
In order to prove  polynomial complexity of II-QIPM we need to find a positive lower bound for the step length $\hat{\alpha}^k$. The following lemma is bounding some remaining elements to get the step length bound.
\begin{lemma}\label{lemma: bound on complementarity}
There exist $0\leq\nu^k=\Ocal(n^2\mu^k)$ such that $\left|\Delta x^k_i \Delta s^k_i -\gamma_1 \frac{(\Delta  x^k)^T\Delta s^k}{n}\right|\leq \nu^k$ for $i\in\{1,2,\dots,n\}$ and $|(\Delta  x^k)^T\Delta s^k|\leq \nu^k$.
\end{lemma}
\begin{proof}
To prove this lemma, we need to do the following steps
\begin{enumerate}
    \item[]Step 1. finding bound $\Ccal_1^k=\Ocal(n\mu^k)$ such that $\omega\theta^{k-1}\|x^k,s^k\|_1\leq \Ccal_1^k$;
    \item[]Step 2. finding bound $\Ccal_2^k=\Ocal (n \sqrt{\mu^k})$ such that $\|D^{-1}\Delta x^k\|\leq \Ccal_2^k$ and $\|D\Delta s^k\|\leq \Ccal_2^k$;
    \item[]Step 3. finding bound $0\leq\nu^k=\Ocal(n^2\mu^k)$ such that 
    $$|(\Delta  x^k)^T\Delta s^k|\leq \nu^k \text{ and }\left|\Delta x^k_i \Delta s^k_i -\gamma_1 \frac{(\Delta  x^k)^T\Delta s^k}{n}\right|\leq \nu^k \text{ for }i\in\{1,2,\dots,n\}.$$
\end{enumerate}
\textbf{Step 1.}
Let us define $\theta^k=\prod^{k}_{i=0}(1-\hat{\alpha}^i)$. One can verify that
\begin{equation}\label{eq: thetaineq}
    R_P^k=\theta^{k-1}R_p^0,\text{ and }
    R_D^k=\theta^{k-1}R_D^0.
\end{equation}
Based on the definition of $\Ncal(\gamma_1,\gamma_2)$ and the choice of $\gamma_2$, we have
\begin{equation}\label{eq: residual rel}
    \frac{\|(R_P^k,R_D^k)\|}{\mu^k}\leq\frac{\|(R_P^0,R_D^0)\|}{\mu^0},
\end{equation}
implying $\mu^k\geq\theta^{k-1}\mu^0$. We also define
\begin{align*}
    (\xbar^k,\ybar^k,\sbar^k)=&\theta^{k-1}(x^0,y^0,s^0)+(1-\theta^{k-1})(x^*,y^*,s^*)-(x^k,y^k,s^k),
\end{align*}
where, $(x^*,y^*,s^*)\in\Pcal\Dcal^*$. 
One can verify that 
\begin{align*}
    A^T\ybar^k+\sbar^k&=0,\\
    A\xbar^k&=0.
\end{align*}
Since $\sbar^k$ is in the row space of $A$ and $\xbar^k$ is in the null space of $A$, we have $(\xbar^k)^T\sbar^k=0$, or equivalently,
\begin{align*}
    [\theta^{k-1}x^0+(1-\theta^{k-1})x^*-x^k]^T[\theta^{k-1}s^0+(1-\theta^{k-1})s^*-s^k]&=0.
\end{align*}
Since $(x^*,s^*,x^k,s^k)\geq 0$, $x^*s^*=0$, and $(x^0,s^0)=(\omega e, \omega e)$, we can write
\begin{align}
    [\theta^{k-1}x^0+(1-\theta^{k-1})x^*-x^k]^T[\theta^{k-1}s^0+(1-\theta^{k-1})s^*-s^k]&=0,\\
    \scalebox{.95}{$
    (\theta^{k-1})^2(s^0)^Tx^0+\theta^{k-1}(1-\theta^{k-1})[(s^0)^Tx^*+(s^*)^Tx^0]+(s^k)^Tx^k$}&\geq\theta^{k-1}[(s^0)^Tx^k+(x^0)^Ts^k], \\
    (\theta^{k-1})^2n\mu^0+2\theta^{k-1}(1-\theta^{k-1})n\mu^0+(s^k)^Tx^k&\geq\theta^{k-1}\omega[(e)^Tx^k+(e)^Ts^k],\label{eq: t3}\\
    \theta^{k-1}n\mu^k+2(1-\theta^{k-1})\mu^k+n\mu^k&\geq\theta^{k-1}\omega\|x^k,s^k\|_1.\label{eq: 12}
\end{align}
Inequality \eqref{eq: t3} is obtained by using $(x^0,s^0)=(\omega e,\omega e)$ and $\|(x^*,s^*)\|_\infty\leq \omega$, and the last inequality is obtained by using \eqref{eq: thetaineq}.
Thus, let $\Ccal_1^k$ be defined as the left-hand side of \eqref{eq: 12}, we have $\Ccal_1^k=\Ocal(n\mu^k)$.

\textbf{Step 2.} 
In this step we define
\begin{align*}
    (\xbar^k,\ybar^k,\sbar^k)=&(\Delta x^k,\Delta y^k,\Delta s^k) +\theta^{k-1}(x^0,y^0,s^0)-\theta^{k-1}(x^*,y^*,s^*),
\end{align*}
where, $(x^*,y^*,s^*)\in\Pcal\Dcal^*$. Similar to Step 1, one can verify that 
\begin{align*}
    A^T\ybar^k+\sbar^k&=0,\\
    A\xbar^k&=0,\\
    (\xbar^k)^T\sbar^k&=0.
\end{align*}
Consequently, one can verify
\begin{align*}
&\|D^{-1}(\Delta x^k +\theta^{k-1}(x^0-x^*))+D(\Delta s^k +\theta^{k-1}(s^0-s^*))\|\\ 
&\qquad =\|D^{-1}(\Delta x^k +\theta^{k-1}(x^0-x^*))\|+\|D(\Delta s^k +\theta^{k-1}(s^0-s^*))\|\\ 
&\qquad \leq \|(XS)^{-1/2}\|(\|XSe-\beta_1\mu^ke\|+\eta\mu^k) +\theta^{k-1}\|D^{-1}(x^0-x^*))\|+\theta^{k-1}\|D(s^0-s^*)\|.
\end{align*}
Now, we have
\begin{align*}
&\|D^{-1}\Delta x^k\| \leq\|(XS)^{-1/2}\|(\|XSe-\beta_1\mu^ke\|+\eta\mu^k) +2\theta^{k-1}\|D^{-1}(x^0-x^*))\|+2\theta^{k-1}\|D(s^0-s^*)\|.
\end{align*}
According to pages 116-118 of \cite{Wright1997_Primal}, we can derive the following inequalities: 
\begin{align*}
    \|(XS)^{-1/2}\|&\leq \frac{1}{\sqrt{\gamma_1}\sqrt{\mu^k}},\\
    \|XSe-\beta_1\mu^ke\|&\leq n \mu^k,\\
    \theta^{k-1}\|D^{-1}(x^0-x^*))\|+\theta^{k-1}\|D(s^0-s^*)\|&\leq \theta^{k-1}\|x^k,s^k\|_1\|(XS)^{-1/2}\|\max(\|x^0-x^*\|,\|s^0-s^*\|)\\
    &\leq\frac{\omega \theta^{k-1}\|x^k,s^k\|_1}{\sqrt{\gamma_1}\sqrt{\mu^k}}\leq \frac{C_1}{\sqrt{\gamma_1}\sqrt{\mu^k}}.
\end{align*}
Thus, $\|D^{-1}\Delta x^k\|\leq\frac{n\mu^k+\eta\mu^k+\Ccal_1^k}{\sqrt{\gamma_1}\sqrt{\mu^k}} = \Ccal_2^k$, where $ \Ccal_2^k=\Ocal(n\sqrt{\mu^k})$. We can similarly show that $\|D\Delta s^k\|\leq \Ccal_2^k$. 

\textbf{Step 3.}
Based on Steps 1 and 2, we have
\begin{align*}
    (\Delta x^k)^T\Delta s^k&\leq \|D^{-1}\Delta x^k\|\|D\Delta s^k\| \leq (\Ccal_2^k)^2,\\
    |\Delta x^k_i\Delta s^k_i|&\leq \|D^{-1}\Delta x^k\|\|D\Delta s^k\| \leq (\Ccal_2^k)^2,\\
   \left|\Delta x^k_i \Delta s^k_i -\gamma_1 \frac{(\Delta  x^k)^T\Delta s^k}{n}\right|&\leq (1+\frac{\gamma_1}{n})(\Ccal_2^k)^2\leq  2 (\Ccal_2^k)^2.
\end{align*}
Thus, $\nu^k=2 (\Ccal_2^k)^2=\Ocal(n^2\mu^k)$.
\end{proof}
In the next lemma, we present a strictly positive lower bound for $\hat{\alpha}^k$. In the proof, we use parameters $\delta_1$,  $\delta_2$, and $\delta_3$ defined as follows:
\begin{equation}\label{eq:parameters}
\begin{aligned}
\delta_1=\frac{(1-\gamma_1)(\beta_1-\eta)}{n}>0,\quad
\delta_2=\beta_1-\eta>0, \quad
\delta_3=\beta_2-\beta_1+\eta>0.
\end{aligned}
\end{equation}

\begin{lemma} \label{lemma: bound of alpha}
At line \ref{alg-step: alpha_hat^k} of Algorithm~\ref{alg: II-QIPM}, at iteration $k$, we have 
\begin{equation*}
\hat{\alpha}^k\geq \tilde{\alpha}^k\coloneqq \min \left\{1,\  \ \min  \{\delta_1,\delta_2,\delta_3\}\frac{(x^k)^T s^k}{\nu^k}\right \}>0.
\end{equation*}
\end{lemma}
\begin{proof}
It is enough to show that the conditions of Lemma \ref{lem: neigh} hold for all $\alpha \in [0,\tilde{\alpha}^k]$.
Based on Lemma~\ref{lemma: iterations properties}, for any $\alpha \in [0,\tilde{\alpha}^k]$, we have
\begin{subequations}
\begin{alignat}{2}
    G^k_i(\alpha)&=(x^k_i+\alpha \Delta x^k_i)(s^k_i+\alpha \Delta s^k_i)- \gamma_1\frac{(x^k+\alpha \Delta x^k)^T(s^k+\alpha \Delta s^k)}{n} \label{eq: alpha_tilde_f a}\\
    &= \textstyle (1-\alpha)x^k_is^k_i+\alpha (\beta_1-\eta) \mu^k+\alpha^2 \Delta x^k_i\Delta s^k_i-\gamma_1\frac{(1+\alpha(\beta_1-\eta-1))(x^k)^Ts^k+\alpha^2(\Delta  x^k)^T\Delta s^k}{n}\label{eq: alpha_tilde_f b}\\
    & \geq  \scalebox{0.95}{$\alpha^2 (\Delta x^k_i\Delta s^k_i-\frac{\gamma_1}{n}(\Delta  x^k)^T\Delta s^k )+(1-\alpha) (x^k_is^k_i-\frac{\gamma_1}{n}(x^k)^Ts^k)  +\alpha (\beta_1-\eta) (1-\gamma_1)\mu^k $}\label{eq: alpha_tilde_f c}\\
    &\geq -\alpha^2\nu^k +\alpha \delta_1 (x^k)^Ts^k \geq \alpha\big(\delta_1(x^k)^Ts^k-\nu^k \tilde{\alpha}^k\big)\geq 0. \label{eq: alpha_tilde_f d} 
\end{alignat}
\end{subequations}
Equality \eqref{eq: alpha_tilde_f b} follows from equation~\eqref{eq:primal relationship} of Lemma \ref{lemma: iterations properties}, and inequality \eqref{eq: alpha_tilde_f d} is due to the definition of $\tilde{\alpha}^k$ and the neighborhood $\mathcal{N}(\gamma_1,\gamma_2)$. Similarly, we show that $ g^k_i(\alpha)\geq 0$ as 
\begin{align*}
    g^k_i(\alpha)&=x^k(\alpha)^Ts^k(\alpha)- (1-\alpha)(x^k)^Ts^k\\
    &\geq \big(1+\alpha(\beta_1-\eta-1)\big)(x^k)^Ts^k+\alpha^2(\Delta  x^k)^T\Delta s^k- (1-\alpha)(x^k)^Ts^k\\
    & \geq \alpha (\beta_1-\eta)(x^k)^Ts^k -\alpha^2\nu^k\\
    &\geq \alpha\delta_2 (x^k)^Ts^k-\alpha^2\nu^k\\
    &\geq \alpha \big(\delta_2 (x^k)^Ts^k-\nu^k \tilde{\alpha}^k\big)\geq 0.
\end{align*}
Again, by \eqref{eq:complementarity relationship} of Lemma \ref{lemma: iterations properties}, we have for all $\alpha \in [0,\tilde{\alpha}^k]$
\begin{align*}
    h^{k}(\alpha)&=\big(1-\alpha(1-\beta_2)\big)(x^k)^Ts^k-(x^k+\alpha \Delta x^k)^T(s^k+\alpha \Delta s^k)\\
    &= \big(1-\alpha(1-\beta_2)\big)(x^k)^Ts^k- \big(1+\alpha(\beta_1-\eta-1)\big)(x^k)^Ts^k-\alpha^2(\Delta  x^k)^T\Delta s^k\\
    & \geq \alpha (\beta_2 - \beta_1+\eta)(x^k)^Ts^k -\alpha^2\nu^k\\
    &\geq \alpha\delta_3 (x^k)^Ts^k-\alpha^2\nu^k\\
    &\geq \alpha \big(\delta_3 (x^k)^Ts^k-\nu^k \tilde{\alpha}^k\big)\geq 0.
\end{align*}
We showed that for all $\alpha \in [0,\tilde{\alpha}^k]$, all the conditions of Lemma \ref{lem: neigh} hold.
Thus, we can conclude that $\tilde{\alpha}^k\leq \hat{\alpha}^k$ and the proof is complete.
\end{proof}

By Lemma \ref{lemma: bound of alpha}, we have a strictly positive lower bound for step length $\tilde{\alpha}^k$ to remain in the neighborhood of the central path while we decrease the optimality gap.
In what follows,  using the results of the previous lemmas, we establish the iteration complexity of Algorithm \ref{alg: II-QIPM}.

\begin{theorem}\label{theo: convergence}
If Algorithm~\ref{alg: II-QIPM} does not terminate in line \ref{alg-step:check infeasibility}, then it reaches a $\zeta$-optimal solution in at most $\Ocal(n^2\log \frac{1}{\zeta})$ iterations.
\end{theorem}

\begin{proof}
Based on Lemma~\ref{lemma: bound of alpha}, we have
\begin{equation*}
\hat{\alpha}^k \geq \alpha^k \geq \tilde{\alpha}^k \coloneqq 
\min \left\{1,\min\left\{\delta_1,\delta_2,\delta_3\right\}\frac{n\mu^k}{\nu^k}\right\}\in (0,1].
\end{equation*}
Hence, by the definition of the neighborhood, we have
\begin{equation*}
(x^k)^Ts^k\leq \big(1-\tilde{\alpha}^k(1-\beta_2)\big)^k(x^0)^Ts^0.
\end{equation*}
This implies that $\lim_{k\to \infty}(x^k)^Ts^k=0$. 
Thus, the algorithm terminates in finite number of steps.
Since the iterates are in the $\Ncal(\gamma_1,\gamma_2)$ neighborhood of the central path, a $\zeta$-optimal solution is obtained when $\mu^k\leq \frac{\zeta}{\gamma_1}$.
The algorithm stops when 
\begin{equation*} 
\mu^k\leq \big(1-\tilde{\alpha}^k(1-\beta_2)\big)^k\mu^0 \leq \frac{\zeta\mu^0}{\|R_p^0, R_D^0\|}.
\end{equation*}
By the definition of $\tilde{\alpha}^k$ and Lemma \ref{lemma: bound on complementarity}, we have $\frac{1}{\tilde{\alpha}^k}=\Ocal(n^2)$. We can conclude that $k=\Ocal(n^2\log \frac{\omega}{\zeta})$. The proof is complete.
\end{proof}
\begin{remark}\label{cor: converge}
 The sequences of  $\left\{\mu^{k}\right\}$, primal infeasibility $\left\{\|A x^k - b \|\right\}$, and dual infeasibility  $\left\{\|A^T y^k + s^k -c \|\right\}$, generated by II-QIPM,   converge linearly to zero. 
\end{remark}
\begin{remark}\label{cor: exact}
 Based on Theorem \ref{cor: exact sol}, we can calculate an exact solution by a rounding procedure if $ \zeta \leq 2^{-4L}. $ Consequently, the iteration complexity of II-QIPM for finding an exact solution is $\Ocal(n^2L)$.
\end{remark}
In the next section, we provide the total time complexity of II-QIPM. 
%
\subsection{Total time complexity of the II-QIPM}\label{sec: Total time complexity of II-QIPM}
%
As discussed in Lemma  \ref{theo: QLSA solution}, we need to solve the NES at each iteration of the II-QIPM by subsequent application of QLSA and QTA, which requires  $\tilde{\Ocal}_{n,\frac{1}{\mu^k},\|\hat{\sigma}^k\|}(mn+m\sqrt{n}\frac{\kappa_E^k\|\hat{\sigma}^k\|}{\sqrt{\mu^k}})$ computational cost.
We can calculate the total time complexity of the II-QIPM as the product of the complexity of the QLSA and the number of iterations of the II-QIPM.  
The computational cost of the QLSA depends on $\kappa_E^k$, $\|\hat{\sigma}^k\|$, $\|E^k\|_F$, and  $\mu^k$ which change through the algorithm.
In the following theorem, we bound them properly and obtain the detailed total time complexity of the proposed II-QIPM algorithm.
\begin{theorem}
The total time complexity of the proposed II-QIPM with QLSA by \textcite{Chakraborty2018_power} and QTA by \textcite{van2022quantum} is
\begin{equation*}
    \tilde{\Ocal}_{n,\frac{1}{\zeta},\omega,\|\Ahat\|,\|\bhat\|}\bigg(n^2\Big[mn+\frac{m\sqrt{n}\kappa_{\Ahat}(\|\Ahat\|+\|\bhat\|)}{\zeta^3}\Big]\bigg).
\end{equation*}
\end{theorem}

\begin{proof}
To establish the total time complexity of II-QIPM, we need to analyze how the matrices $M^k$ and $E^k$ evolve through the iterations. 
As in \citep{Roos2005_Interior}, considering the optimal partition $B$ and $N$, we have
\begin{equation}\label{eq:diagonal scaling}
    \frac{x_i^k}{s_i^k}=\Ocal(\frac{1}{\mu^{k}})\to\infty\text{  for }i\in B\qquad \text{and} \qquad\frac{x_i^k}{s_i^k}=\Ocal(\mu^{k})\to0\text{  for }i\in N.
\end{equation}
Appropriate bounds are provided in the following results.
\begin{enumerate}[label=(\roman*)]
\item Based on Theorem \ref{theo: convergence}, we have at the termination $\frac{1}{\mu^k}=\Ocal(\frac{1}{\zeta})$ and $\mu^k\leq\Ocal(\mu^0)=\Ocal(\omega^2)$.
\item Since $\|E^k\|\leq \|(D^{k})^{-1}\|\|\Ahat\| \|D^{k}\|$, and $\|D^{k}\|=\Ocal(\frac{1}{\sqrt{\mu^k}}) =\Ocal(\frac{1}{\sqrt{\zeta}})$ by~\eqref{eq:diagonal scaling}. Similarly, $\|(D^{k})^{-1}\|_F=\Ocal(\frac{1}{\sqrt{\zeta}})$, and  we have $\|E^k\|=\Ocal\big(\frac{\|\Ahat\|}{\zeta}\big)$.
Let $\kappa_{\Ahat}$ be the condition number of $\Ahat$. 
Thus, we have $$\kappa^k_E=\Ocal\big(\kappa_{\Ahat}\zeta^{-2}\big).$$
\item In the time complexity of II-QIPM, we also have $\frac{\|\sigma^k\|}{\sqrt{\mu^k}}$ coming from the  precision of QLSA and QTA. We can easily verify that 
\begin{align*}
    \frac{\|\hat{\sigma}^k\|}{\sqrt{\mu^k}}&\leq\frac{\|(D^{k})^{-1}\|}{\sqrt{\mu^k}}(\|\bhat\|+\|\Ahat X^k(S^{k})^{-1}\|\|c-A^Ty^k-s^k\|+\beta_1\|\mu^k \Ahat(S^{k})^{-1}e\|)\\
    &\leq\frac{1}{\mu^k}(\|\bhat\|+\|\Ahat X^k(S^{k})^{-1}\|\|R_D^k\|+\beta_1\|\Ahat(S^k)^{-1}e\|\mu^k).
\end{align*}
Based on the definition of the neighborhood of the central path, $\|R_D^k\|\leq \gamma_2\mu^k$. Thus, we can get bounds $\Ocal(\frac{\|\sigma^k\|}{\sqrt{\mu^k}})=\Ocal(\frac{\|\Ahat\|+\|\bhat\|}{\zeta})$.
\item Based on Lemma~\ref{theo: QLSA solution}, the complexity of QLSA by \textcite{Chakraborty2018_power} and QTA by \textcite{van2022quantum} for building and solving the MNES is 
\begin{equation*}
    \tilde{\Ocal}_{n,\frac{1}{\zeta},\|\Ahat\|,\|\bhat\|}\big(mn+m\sqrt{n}\frac{\kappa_{\Ahat}(\|\Ahat\|+\|\bhat\|)}{\zeta^3}\big).
\end{equation*}
\end{enumerate}
Thus, the detailed time complexity of the proposed II-QIPM with QLSA by \textcite{Chakraborty2018_power} and QTA by \textcite{van2022quantum} is
\begin{equation}\label{eq: detailed complexity of II-QIPM}
    \tilde{\Ocal}_{n,\frac{1}{\zeta},\omega,\|\Ahat\|,\|\bhat\|}\bigg(n^2 \Big[mn+\frac{m\sqrt{n}\kappa_{\Ahat}(\|\Ahat\|+\|\bhat\|)}{\zeta^3}\Big]\bigg).
\end{equation}
The time complexity is achieved by multiplying the number of iterations of II-QIPM and the total cost of each iteration, including building and solving it by QLSA+QTA. Thus, the proof is complete.
\end{proof}
Let $\phi:=\|\Ahat\|+\|\bhat\|$, the complexity of II-QIPM can be simplified as
\begin{equation}\label{eq:final complexity of II-QIPM}
    \tilde{\Ocal}_{n,\frac{1}{\zeta},\omega,\phi}\Big(n^{4}\frac{\phi\kappa_{\Ahat}}{\zeta^{3}}\Big).
\end{equation}
In the complexity of II-QIPM, the $\frac{1}{\zeta}$ factors come from bounding the condition number and from QTA. Based on Theorem~\ref{cor: exact sol}, $\frac{1}{\zeta}=2^{\Ocal(L)}$ leading to exponential complexity. We discuss how we can solve this and improve the complexity of the algorithm by stopping II-QIPM early, e.g., $\zeta=10^{-2}$, and using the Iterative Refinement scheme  discussed in the Section~\ref{sec: iterative refinement} to improve the precision. 

%
\begin{remark}
Some QLSAs, such as \textcite{Harrow2009_Quantum} and \textcite{Childs2017_Quantum}, take advantage of the sparsity of the NES. 
Consider \eqref{eq: modified normal equation}, the sparsity of  $\Mhat^k=\Ahat X^k(S^{k})^{-1}\Ahat^T$ is independent of $X^k$ and $S^k$. So, let $d$ be the maximum number of nonzero elements in any row or column of the matrix $\Ahat \Ahat^T$. 
Since matrix $\Ahat$ has two blocks $\begin{bmatrix}I& A_{\Bhat}^{-1}A_{\Nhat}\end{bmatrix}$, we have $d\leq \min\{m,n-m+1\}$. 
As matrix $\Ahat$ determines the sparsity of matrix $M^k$, so we can take advantage of the sparsity structure of matrix $\Ahat$. In case $\Ahat$ is mostly sparse, e.g., $n-m\ll m$, but has a few dense columns, then this structure can be exploited. As described in \cite{Andersen2004_use}, the sparse part can be separated to solve a sparse linear system by QLSA, and then the use of the Sherman-Morrison-Woodbury \cite{sherman1950adjustment} formula allows calculating the solution of the original linear system efficiently. 
\end{remark}
%
\section{Iterative Refinement Method} \label{sec: iterative refinement}
%
Based on Remark~\ref{cor: exact}, we need $\zeta\leq 2^{-4L}$ to have an exact solution for the LO problem with integer data.
The proposed II-QIPM has exponential complexity in finding an exact solution.
An Iterative Refinement (IR) scheme can be employed to achieve polynomial complexity.
By this scheme,  a series of  LO problems are solved by the II-QIPM with low precision $\hat{\zeta}$, e.g., $\hat{\zeta}=10^{-2}$, and an IR method improves the precision to reach an exact solution.
In the classical IPM literature, the IR method by \textcite{Gleixner2016_Iterative} and the Rational Reconstruction method by \textcite{Gleixner2020_Linear} employed low-precision methods to generate a high-precision solution.
Here, we adopt the IR method to generate a high precision solution that allows the identification of an exact optimal solution. 
While using only low-precision II-QIPM, Theorem \ref{theorem:iterative refinement idea} is the foundation of the IR method.
\begin{theorem}[\textcite{Gleixner2016_Iterative}] \label{theorem:iterative refinement idea}
Let the primal problem $(P)$ be given as \eqref{eq:primal and dual problem}.
For $\tilde{x}\in \mathbb{R}^n$ and $\tilde{y}\in \mathbb{R}^m$ and scaling factor $\nabla>1$ consider the refining problem $(\bar{P})$
\begin{equation*}
\min \left\{ \nabla \bar{c}^Tx| Ax= \nabla \bar{b} \text{ and } x\geq -\nabla\tilde{x}\right\},
\end{equation*} 
where $\bar{c}= c -A^T \tilde{y}$ and $\bar{b}=b-A\tilde{x}$. 
Then $\bar{x}$ and $\bar{y}$ are $\hat{\zeta}$-optimal solution for problem $(\bar{P})$ if and only if $\tilde{x}+\frac{1}{\nabla}\bar{x}$  and $\tilde{y}+\frac{1}{\nabla}\bar{y}$ are the $\frac{\hat{\zeta}}{\nabla}$-optimal solution for problem $(P)$.
\end{theorem}

As presented in Algorithm \ref{alg: iterative refinement II-QIPM}, we use the II-QIPM of Algorithm~\ref{alg: II-QIPM} to solve the refining model and update the solution with an intelligent scaling procedure. 
Theorem \ref{theorem: IR complexity} shows that a polynomial number of iterations are sufficient to reach an exact optimal solution.
\begin{algorithm}[ht]
	\caption{IR-II-QIPM} \label{alg: iterative refinement II-QIPM}
	\begin{algorithmic}[1]
		\Require $\big(A\in \mathbb{Z}^{m\times n}, b\in \mathbb{Z}^{m}, c\in \mathbb{Z}^{n}\big)$ 
		\State Choose scaling multiplier $\rho \in \mathbb{N}$ such that $\rho > 1$
		\State $\hat{\zeta} \gets  10^{-2},\zeta \gets 2^{-4L}$
		\State  $k \gets 0$
		\State $(x^{*0}, y^{*0}, s^{*0}) \gets$ Algorithm~\ref{alg: II-QIPM} $(A,b,c)$ with $\hat{\zeta}$ precision
		\While{$r^k>\zeta$}
		\State $\pi^k\gets \max \left\{r^{k},\frac{1}{\rho\nabla^{k}}\right\}$
		\State $\nabla^k \gets 2^{\left\lceil\log(\sfrac{1}{\pi_{k}})\right\rceil}$
		\State $\bar{b}^{k} \gets b-Ax^{*k}$ and $\bar{c}^{k} \gets c-A^Ty^{*k}$
		\State $(\hat{x}^{k}, \hat{y}^{k}, \hat{s}^k) \gets$ Algorithm~\ref{alg: II-QIPM} $(A,\nabla^k\bar{b}^{k},\nabla^k\bar{c}^{k})$ with $\hat{\zeta}$ precision
		\State $x^{*k+1}\gets x^{*k}+\frac{1}{\nabla^k} \hat{x}^k$ and $y^{*k+1}\gets y^{*k}+\frac{1}{\nabla^k}\hat{y}^k$
		\State $r^{k+1} \gets \max \left\{\max_j |\bar{b}_{\hspace{1pt}j}^{k+1}|, \max_i (-\bar{c}_{\hspace{1pt}i}^{k+1} ), \sum_{i} |\bar{c}_{\hspace{1pt}i}^{k+1} x^{*k+1}_i| \right\}$
		\State $k \gets k+1$
		\EndWhile
	\end{algorithmic}
\end{algorithm}
\begin{theorem} \label{theorem: IR complexity}
Let $(A,b,c)$ be integer. The number of iterations of Algorithm \ref{alg: iterative refinement II-QIPM} to get $2^{-4L}$-precise solution is at most $\Ocal(L)$.
\begin{proof}
Based on Corollary 3.6 in \cite{ Gleixner2016_Iterative}, we have
$$
 \left\lceil \frac{\log (\zeta)}{\log(\hat{\zeta})} \right\rceil=\left\lceil \frac{-4L\log(2)}{-2} \right\rceil=\Ocal(L).
$$
\end{proof}

\end{theorem}
One can observe that, except $\omega$, all parameters in the complexity of II-QIPM as specified in \eqref{eq:final complexity of II-QIPM}, are constant in all iterations of the IR method. 
Lemma~\ref{lemma: bound on omega} provides an upper bound for $\omega^k$ at iteration $k$ of the IR  method. 
%
\begin{lemma} \label{lemma: bound on omega}
At the  $k$\textsuperscript{th} iteration of the IR method, let $(\hat{x}^{*k},\hat{y}^{*k},\hat{s}^{*k})$ be an exact optimal solution of the refining problem $(A,\nabla^k\bar{b}_{k-1},\nabla^k\bar{c}_{k-1})$, and $\omega^k\geq\max\{1,\|\hat{x}^{*k},\hat{s}^{*k}\|_{\infty}\}$. Then, $\omega^k=\Ocal\big((2\rho)^L\big)$.
\end{lemma}
%
\begin{proof}
From Theorem \ref{theorem: IR complexity} and last line of Algorithm \ref{alg: iterative refinement II-QIPM}, we have
\begin{align*}
    x^*&= x^*_{0}+\sum_{k=1}^{\Ocal(L)}\frac{1}{\nabla^k} \hat{x}^k,\text{ and }s^*= s^*_{0}+\sum_{k=1}^{\Ocal(L)}\frac{1}{\nabla^k} \hat{s}^k.
\end{align*}
Based on Lemma~\ref{lemma: L bound}, we know that $\|x^*,s^*\|_{\infty}\leq 2^L$, then we have
$\|\hat{x}^{*k},\hat{s}^{*k}\|_{\infty}\leq \nabla^k2^L.$

Based on the procedure of updating the scaling factor in Algorithm \ref{alg: iterative refinement II-QIPM}, we can drive $\nabla^k=\Ocal(\rho^k)$. We can conclude that $\omega^k=\Ocal\big((2\rho)^L\big)$.
\end{proof}
%
In Theorem \ref{theorem: final complexity}, we have the total time complexity of the IR method using the proposed II-QIPM to find an exact optimal solution for LO problems.
%
\begin{theorem}\label{theorem: final complexity}
Let $\hat{\zeta} =10^{-2}$, then the total time complexity of finding an exact optimal solution using the IR-II-QIPM Algorithm \ref{alg: II-QIPM} for solving the LO problem \eqref{eq:primal and dual problem} is polynomial with 

\begin{equation*}
\tilde{\Ocal}_{n,\|\Ahat\|,\|\bhat\|}\bigg(n^2L\Big[mn+m\sqrt{n}\kappa_{\Ahat}(\|\Ahat\|+\|\bhat\|)\Big]\bigg),
\end{equation*}
the arithmetic operations, where $\hat{A}$ and $\hat{b}$ are preprocessed $A$ and $b$.
\end{theorem}

\begin{proof}
The proof follows from combining the result of Theorem \ref{theorem: IR complexity}, the total time complexity of the proposed II-QIPM in~\eqref{eq: detailed complexity of II-QIPM}, and Lemma~\ref{lemma: bound on omega}.
\end{proof}
%
\begin{corollary}
 The simplified complexity of the proposed IR-II-QIPM using the \eqref{eq: modified normal equation} is $$\tilde{\Ocal}_{n,\phi}(n^{4}L\phi\kappa_{\Ahat}).$$ 
 For finding $\zeta$-optimal solution the complexity of IR-II-QIPM is
 $$\tilde{\Ocal}_{n,\phi,\omega,\frac{1}{\zeta}}(n^{4}\phi\kappa_{\Ahat}).$$
\end{corollary}
Thus, the iterative refinement procedure speeds QIPM up exponentially with respect to precision. It also addresses the condition number since the growing  condition number of the Newton system is replaced by the condition number of constant matrix $\Ahat$. 

The following section presents the results of our numerical results.

\section{Numerical Experiments} \label{sec: Numerical Experiments}

This section provides numerical results for the proposed II-QIPM using the QISKIT AQUA quantum simulator. 
Due to the limited number of qubits available in quantum computers and simulators, we use the NES, which has a smaller dimension.
The numerical results are run on a workstation with Dual Intel Xeon{\textregistered} CPU E5-2630~@ 2.20 GHz (20 cores) and 64 GB of RAM. 
For the computational experiments, we have developed a Python \texttt{qipm} package available for public use\footnote{\url{https://github.com/qcol-lu/qipm}}.

IBM has implemented a QLSA, which is similar to the HHL method, without block-encoding and QRAM.
With the current technology, the number of available qubits in gate-based quantum computers is limited to about one hundred.
One of the main issues with quantum computers is that they are not scalable compared to classical computers. 
Currently, larger NISQ\footnote{NISQ: Noisy Intermediate Scale Quantum} devices suffer more from the lack of precision. 
On the other hand, quantum simulator algorithms are computationally expensive. 
The maximum number of qubits in a quantum simulator is roughly similar to that in an actual quantum computer.
The main advantage of using a quantum simulator is that we do not need to handle the noise of NISQ devices. 
Despite this, we still need to deal with a high error level due to insufficient qubits and the high cost of QLSA+QTA in order to find high-quality solutions.

We used two post-processing procedures to improve the performance of the QLSA.
\begin{enumerate}[label=(\roman*)]
    \item We first scale the linear system solution such that it satisfies the  equality~$ \|Mz\|=\|\sigma\|$. 
    \item We also check the sign of the linear system solution by comparing it with its negate.
\end{enumerate}
The condition number increases the solution time.
Notably, the dimension of the linear system increases the solution time as a step function for building the quantum circuit.
The dimension of the system must be a power of two.
The simulator expands the system size to the smallest possible power of two.
Table~\ref{tab: condition number and size of quantum circuit} presents the number of qubits in quantum circuits for achieving the same precision.
The QISKIT simulator error oscillates between zero and the norm of an actual solution.  
In some cases, the IBM QISKIT simulator fails, and it reports a zero vector as a solution. 
While the simulator has a parameter for tuning precision, it violates the predefined precision.
We could not find any meaningful relationship between the precision of the simulator and the condition number of the coefficient matrix, and the dimension of the system.
In the following sections, we discuss the implementation of the proposed II-QIPM and IR-II-QIPM using the QISKIT simulator of QLSA and evaluate their performance.
\begin{table}[h]
\centering
\begin{tabular}{r|c c c c c c c c c c c }
\hline
$\kappa$ &  $2^1$ & $2^2$ & $2^3$ & $2^4$ &$2^5$ &$2^6$ &$2^7$ & $2^8$& $2^9$ & $2^{10}$ & $2^{11}$\\ \hline
   \#qubits &   6  & 6 & 7 & 8 & 9  & 10 & 11 & 12 & 12 & 13 & 15\\\hline
\end{tabular}
\caption{Size of the circuit for linear systems with different condition numbers.}
\label{tab: condition number and size of quantum circuit}
\end{table}

\subsection{Evaluation of the II-QIPM} \label{sec: num II-QIPM}
%
As discussed in Section~\ref{sec: Quantum Linear Algebra}, QLSAs have better dependence on the size of the linear system than classical algorithms.
However, unreliable qubits cause a significant error in the solution of a linear system.
Here, we use the IBM QISKIT simulator to solve linear systems arising in the proposed II-QIPM.
For a fair comparison, we only consider those experiments that reach the desired precision.  
The running time of the quantum computers and quantum simulators is not comparable.
Thus, instead of running time, we use the number of iterations as a performance measure.

We use the random instance generator of \cite{generator}, where the norm of primal and dual solutions is set to two. 
The norms of the coefficient matrix and the RHS vector are set to one and two, respectively.
The condition number of the coefficient matrix is set to two.
The desired precision of the II-QIPM is equal to $0.1$.
It is worth noting that we set $\omega$ to 10 and all the instances are feasible.

 Figure~\ref{fig: effect of number of variables} illustrates II-QIPM performance on instances solved to the desired precision.
As illustrated in Figure~\ref{fig: effect of number of variables}, the number of iterations fluctuates with increasing the number of variables.
The noisy behavior of the QLSA and the relatively small-scaled instances justifies this observation.
Figure~\ref{fig: effect of precision} shows the number of iterations increases by increasing the desired precision.
The number of iterations also increases by the norm of the RHS vector, while $\omega$ has less impact on the number of iterations (see Figures~\ref{fig: effect of norm of rhs} and \ref{fig: effect of omega}).
As discussed in Section~\ref{sec: Total time complexity of II-QIPM}, the number of iterations is affected by the error of QLSAs.
Scaling the RHS vector norm increases the error of QLSA and, as a result, increases the number of iterations. 
Furthermore, the larger $\omega$ is, the further away the starting point is from the central optimal solution, i.e., the limit point of the central path.
It implies that increasing $\omega$ should increase the number of iterations.
However, the fast convergence of IPMs and the noisy behavior of the QLSA make it difficult to observe a direct relationship between $\omega$ and the number of iterations.

\begin{figure*}[ht]
\centering
    \begin{subfigure}[t]{0.45\textwidth}            
	\includegraphics[width =1.0 
	\linewidth,trim = {.1cm .1cm .1cm .1cm}, clip] {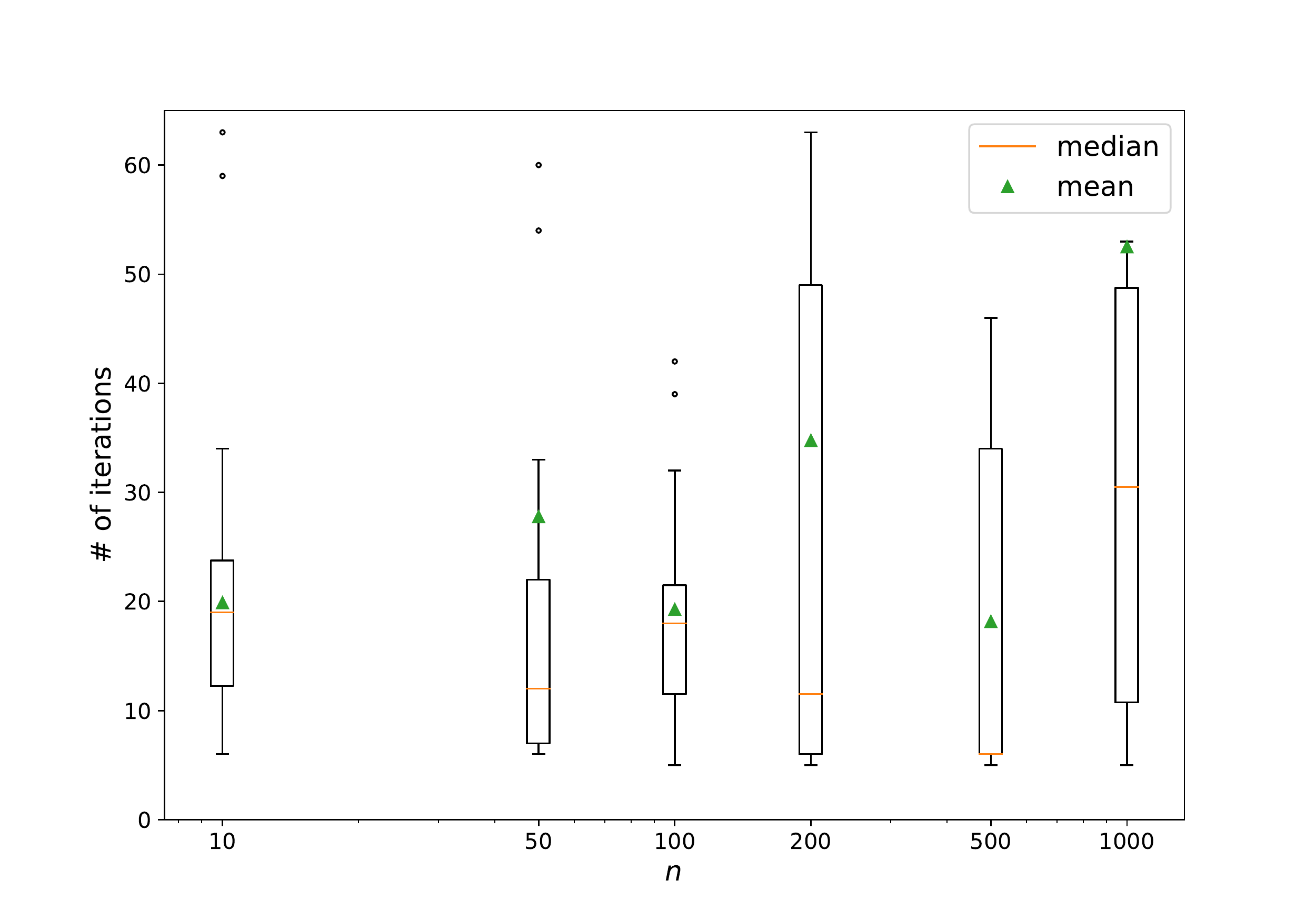}
	\caption{ Effect of number of variables.}
	\label{fig: effect of number of variables}
\end{subfigure}\hspace{3mm}
\begin{subfigure}[t]{0.45\textwidth}
		\centering
		\includegraphics[width =1.0 \linewidth, trim = {.1cm .1cm .1cm .1cm}, clip] {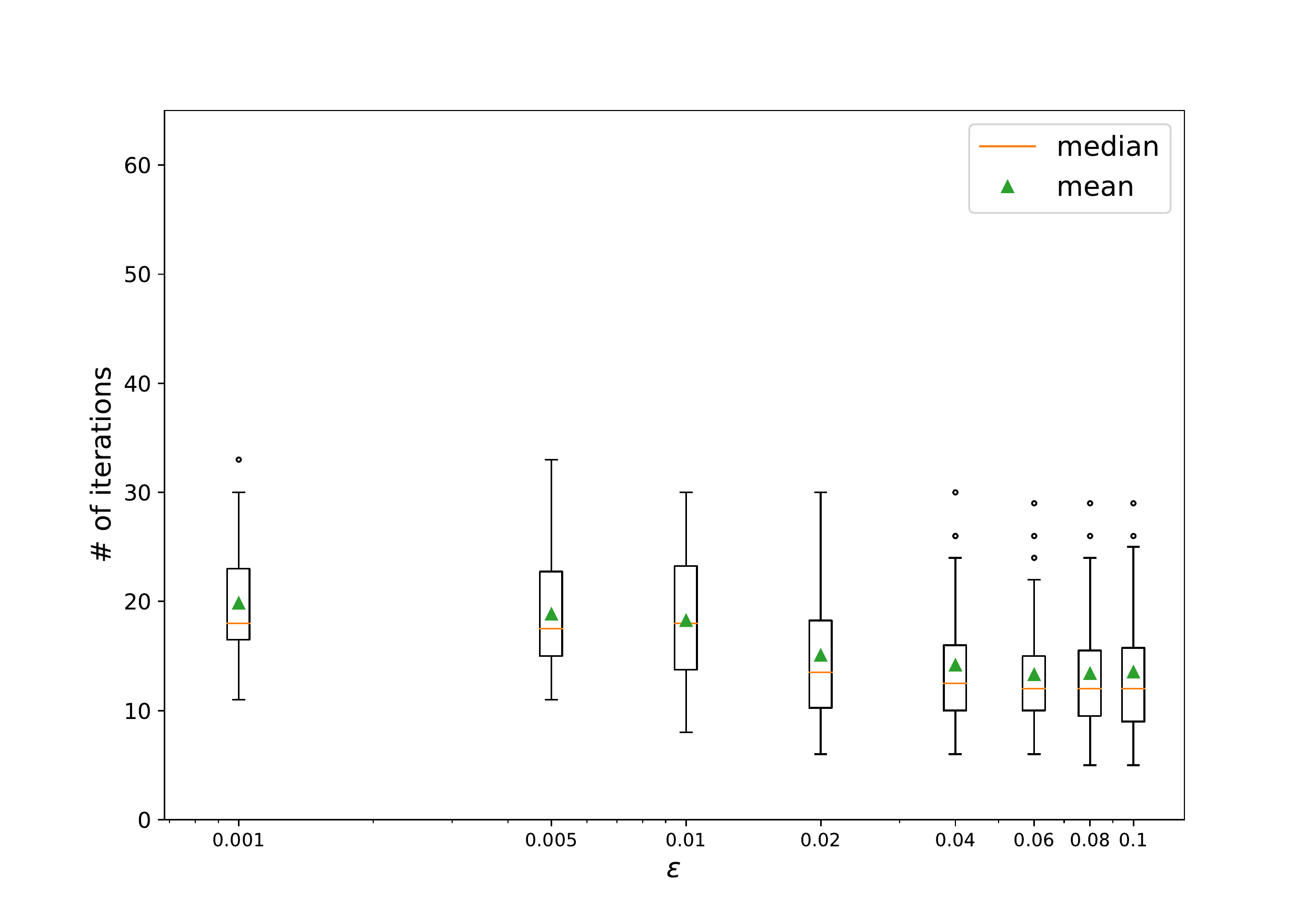}
		\caption{ Effect of the desired precision.}
		\label{fig: effect of precision}
	\end{subfigure}\hspace{3mm}

\begin{subfigure}[t]{0.45\textwidth}
	\centering
\includegraphics[width =1.0 \linewidth,trim = {.1cm .1cm .1cm .1cm}, clip]{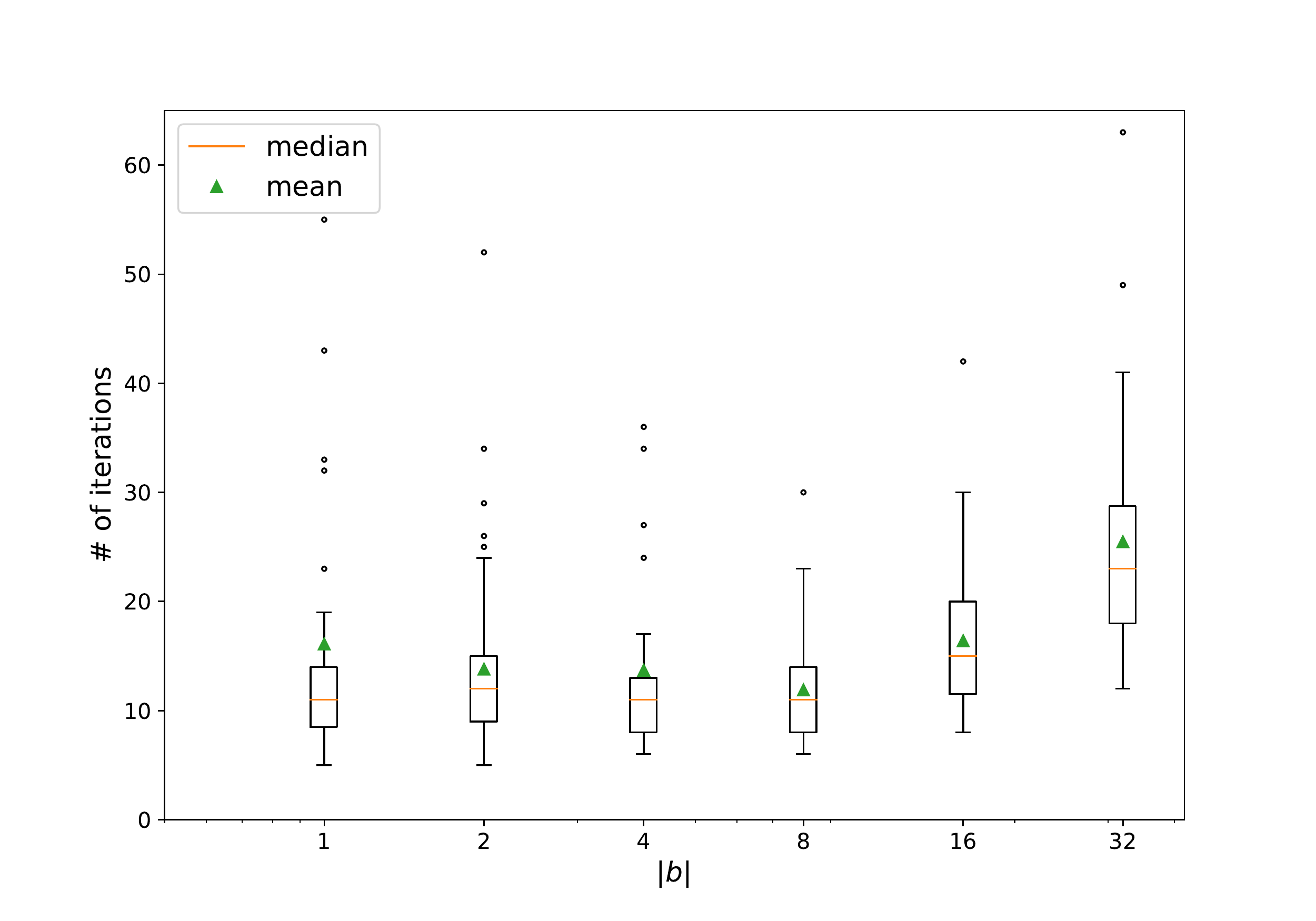}
	\caption{ Effect of norm of RHS vector.}
	\label{fig: effect of norm of rhs}
\end{subfigure}\hspace{3mm}
\begin{subfigure}[t]{0.45\textwidth}
	\centering
\includegraphics[width =1.0 \linewidth,trim = {.1cm .1cm .1cm .1cm}, clip]{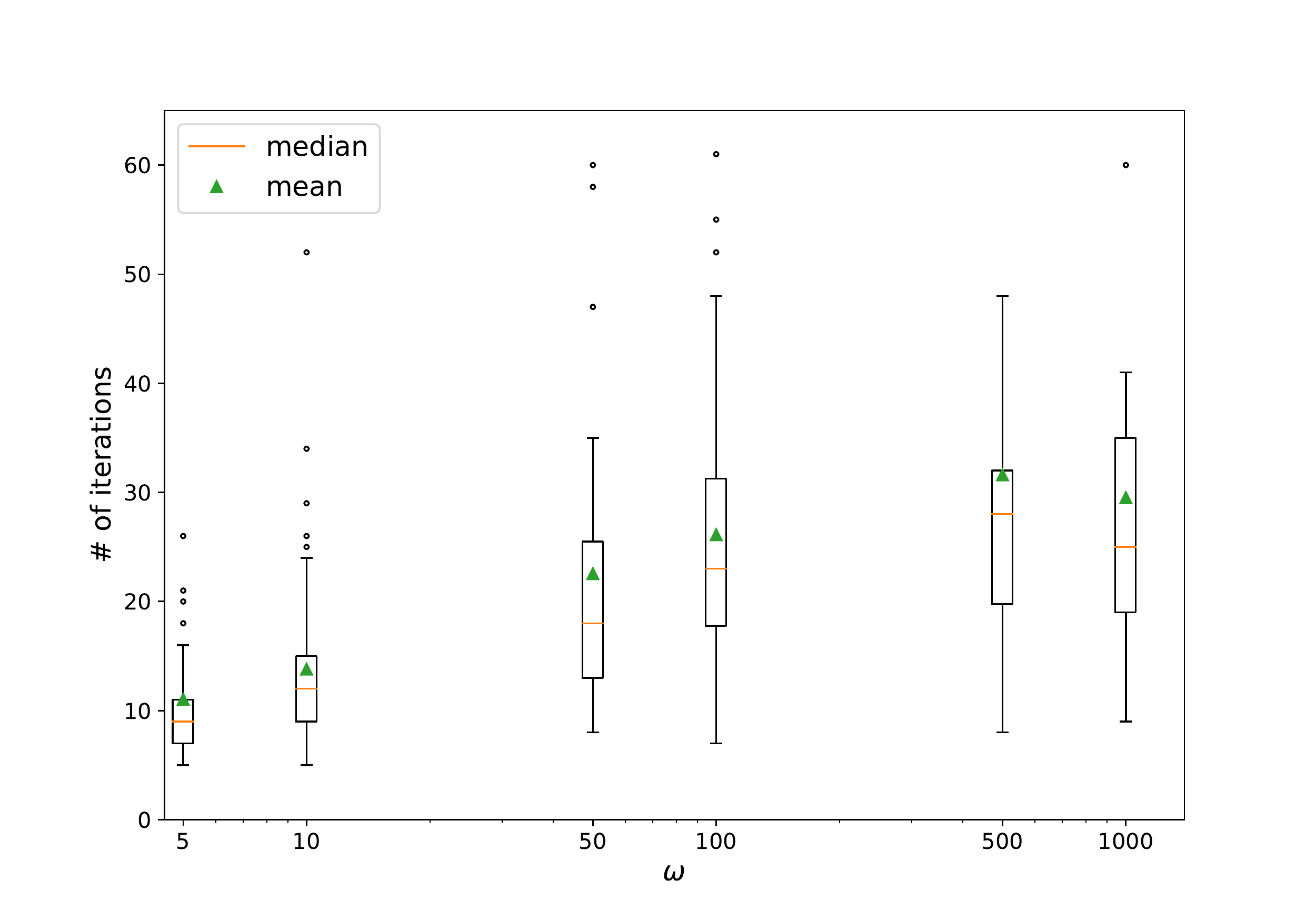}
	\caption{ Effect of $\omega$.}
	\label{fig: effect of omega}
\end{subfigure}\hspace{3mm}
\caption{Effect of different characteristics of a LO problem on the number of iterations.}\label{fig: ii-qipm number of iterations}
\end{figure*}

\subsection{Evaluation of the IR-II-QIPM} \label{sec: num IR-II-QIPM}
Even with smart parameter tuning and perfect implementation, the HHL simulator has limited precision. 
Here, we analyze the performance of the IR method when combined with the II-QIPM. 
One important parameter that highly affects the norm of the RHS vector and, consequently, the HHL simulator's error is $\omega$.
%
%
As the IR progresses, the norm of a solution to the LO subproblem increases.
Thus, we need to set $\omega$ based on the desired precision.
Here, we set $\omega$ to $1000$ to avoid infeasibility reported by the II-QIPM. 

Figure~\ref{fig: Effect of the LO precision on the IR-II-QIPM} shows the logarithmic error of IR-II-QIPM for different desired precisions.
We set the LO precision to 0.01 to evaluate the effect of the IR on the proposed II-QIPM.
Figure~\ref{fig: Effect of iterative refinement on the II-QIPM} illustrates that by using IR, we can reach higher precision.
The final precision in both the pure II-QIPM and the IR-II-QIPM is set to $10^{-4}$.
However, we can still not reach the desired precision in half of the instances because of the QLSA's error. 
\begin{figure}[ht]
	\centering
	\begin{subfigure}[t]{0.48\textwidth}
		\centering
		\includegraphics[width =1.0 \linewidth,trim = {.1cm .1cm .1cm .1cm}, clip] {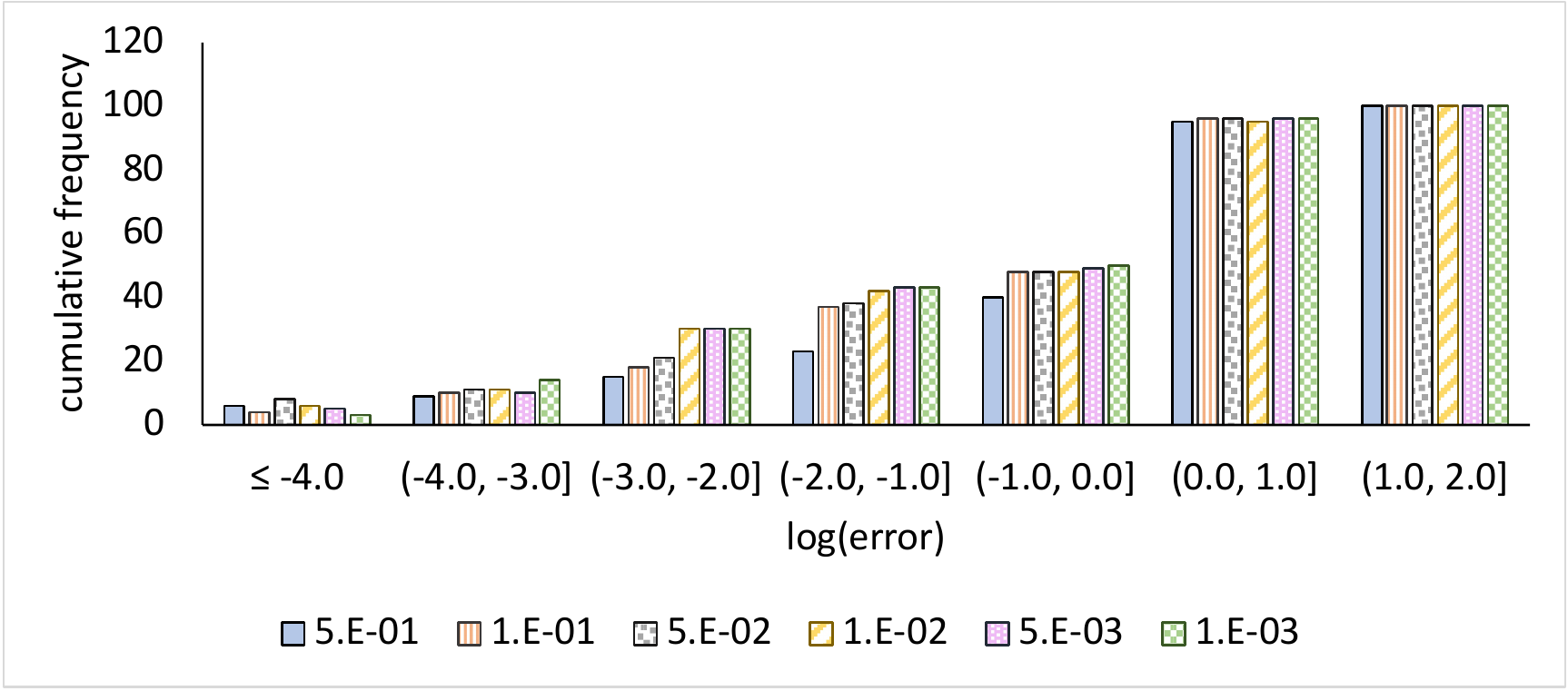}
		\caption{Effect of the LO precision on the IR-II-QIPM.}
		\label{fig: Effect of the LO precision on the IR-II-QIPM}
	\end{subfigure}\hspace{3mm}
	\begin{subfigure}[t]{0.48\textwidth}
		\centering
		\includegraphics[width =1.0 \linewidth, trim = {.1cm .1cm .1cm .1cm}, clip] {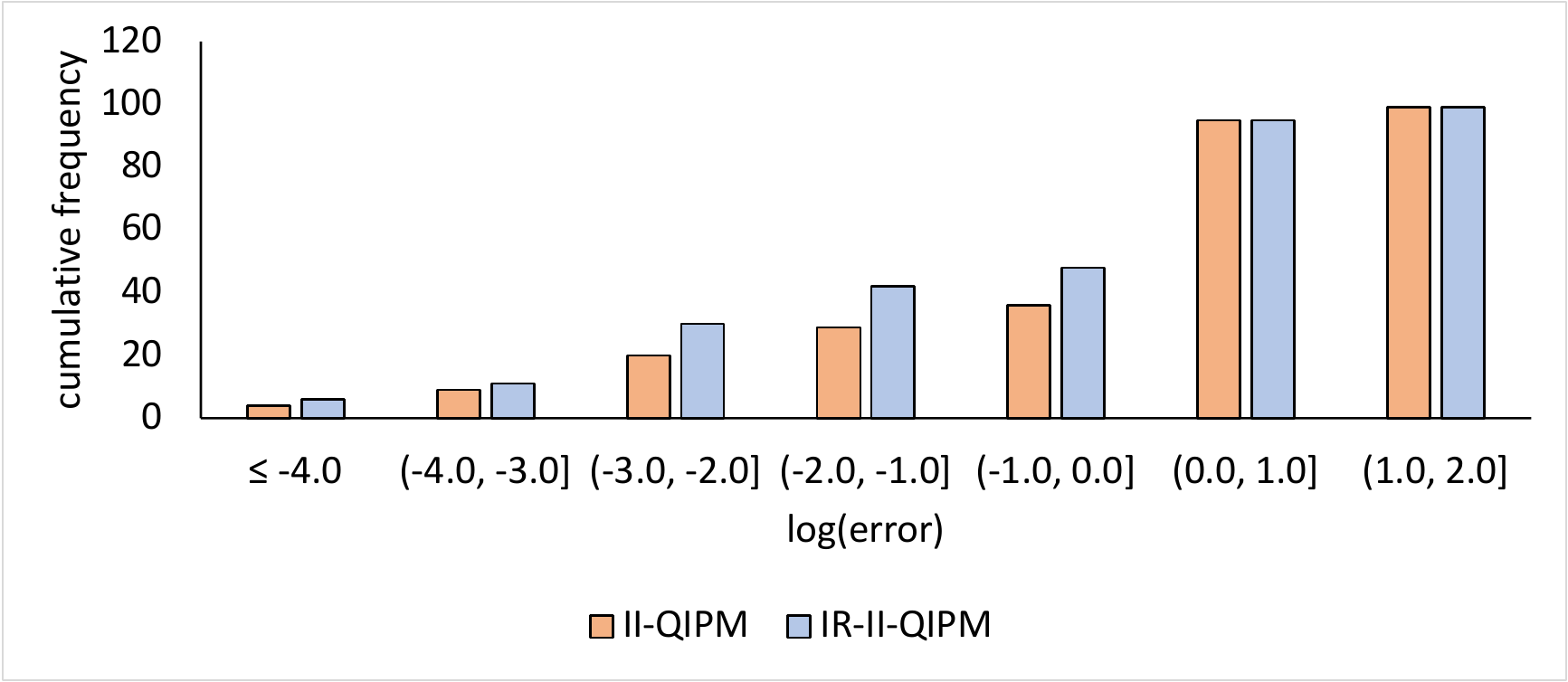}
		\caption{Effect of iterative refinement on the II-QIPM.}
		\label{fig: Effect of iterative refinement on the II-QIPM}
	\end{subfigure}
	\caption{Obtained precision versus the desired precision on the IR-II-QIPM.}
\end{figure}
By reducing the condition number, the iterative refinement method speeds up the solution time of Newton systems.
The condition number of the Newton system is bounded by $\Ocal(\kappa_{\Ahat})$. 
Figure~\ref{fig: conIR} shows how the condition number of the solved linear systems in IR-QIPM is bounded, while the condition number of the solved linear systems in QIPM without IR goes to infinity.
\begin{figure}[H]
    \centering
    \includegraphics[scale=0.55]{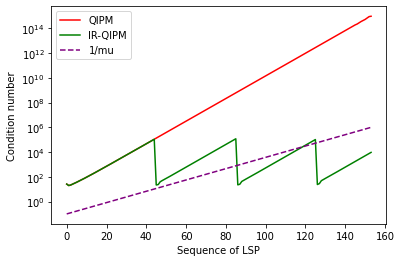}
    \caption{The condition number of linear systems in QIPM and IR-QIPM to get $10^{-6}$-precision solution for a primal-degenerate LO with 10 variables and 5 constraints}
    \label{fig: conIR}
\end{figure}

%
\section{Conclusion}\label{sec: conclusion}
This paper analyzes in detail the benefits and challenges of novel proposed Quantum Interior Point Methods.
Specifically, we analyze the use of QLSAs within IPMs and present a convergent II-QIPM.
Previous papers overlooked that when one uses QLSA with QTA, the solution of the Newton system is inexact, and the Newton system's condition number goes to infinity as IPMs approach the optimal set. 
Here, we also adopt the IR method to find an exact solution in polynomial time. 
After addressing issues in earlier QIPMs, we proved the correctness and convergence of the proposed II-QIPM and analyzed its performance, both theoretically and empirically.
 \begin{table}[ht]
    \centering
    \begin{tabular}{ |c|c|l| } 
 \hline
 Algorithm & 

 Simplified Complexity& Comment \\ 
 \hline
II-IPM with Cholesky  & 
$\Ocal\big(n^5L\big)$&\\
IR-II-IPM with CG&
$\Ocal\big(n^5L\kappa_{\Ahat}\big)$& Anticipated complexity and needing complete analysis.\\
QIPM of \citep{kerenidisParkas2020_quantum} & 
$\tilde{\Ocal}_{n}\big(n^{2}L\kappa^3 2^{4L}\big)$& Unattainable due to using exact IPM complexity.\\
QIPM of \citep{Casares2020_quantum}& 
$\tilde{\Ocal}_{n}\big(n^2L \kappa^2 2^{2L}\big)$& Unattainable due to using exact IPM complexity.\\
Proposed IR-II-QIPM&
$\tilde{\Ocal}_{n,\phi}\big(n^4L\phi\kappa_{\Ahat}\big)$&\\

\hline
\end{tabular}
    \caption{Time complexity of finding an exact solution}
    \label{tab: final complexity}
\end{table}

Table~\ref{tab: final complexity} compares the complexity result of the proposed IR-II-QIPM with an analogous classical II-IPM and two recent QIPMs.
The proposed IR-II-QIPM has polynomial complexity, while the other QIPMs cannot find an exact optimal solution in polynomial time. 
The exponential complexity of those QIPMs is caused by QLSA's error and the increasing condition number of the Newton system. 
At first glance, one can get the impression that the other two QIPMs have better complexity with respect to dimension. 
Still, these time complexities cannot be attained since they only contain the iteration complexity of exact IPMs, while these QIPMs solve the Newton system inexactly. 
They also need appropriate bounds for condition number $\kappa$ and precision $\epsilon$ based on their setting of QIPMs.
To correct the complexity of the QIPMs, at least $\Ocal(n^{1.5})$ must be added for inexact Newton steps and an appropriate upper bound for QLSAs' error. 

The complexity of the proposed IR-II-QIPM has better dependence on $n$ than both its classical counterparts and the realistic complexity of its quantum counterparts.
Still, the complexity of the proposed method depends on constants $\kappa_{\Ahat}$ and $\|\Ahat\|$.
One may apply scaling and preconditioning techniques for LO problems with large $\kappa_{\Ahat}$ to decrease $\kappa_{\Ahat}$.
\textcite{Monteiro2003_Convergence} used the speculated optimal partition instead of predefined basis $(\Bhat,\Nhat)$ as a preconditioned NES. 
The major problem of this approach is that the cost of calculating the precondition $A_{\Bhat}^{-1}$ in each iteration will destroy the quantum speed up in QIPMs. 
A viable research direction is to explore how to mitigate the effects of condition number and norm of Newton systems with a quantum-friendly approach.
In addition, the proposed iterative refinement approach can also be used to mitigate the impact of the condition number in classical IPMs with CG in solving the NES. 
The complexity analysis of Inexact IPM augmented with an iterative refinement method that uses classical iterative solvers for solving the Newton system is part of ongoing work.
However, the anticipated complexity of such algorithms, as reported in Table~\ref{tab: final complexity}, will have unfavorable dependence on dimension than the proposed IR-II-QIPM.
Because at each iteration, there are some matrix-matrix products to build the Newton system. 
By using QLSA+QTA, we can avoid the costs of classical matrix-matrix products and achieve polynomial quantum speedups.

We analyzed how the state-of-the-art QLSA+QTA can solve a classical linear system problem with quadratic dependence on dimension and linear dependence on the condition number and precision.
Although there are classical polynomial-time algorithms to solve LSPs in general form, QLSA-QTA is more scalable with regard to dimension than classical iterative and direct approaches.
The performance of QIPMs will be improved if faster QTAs and QLSAs are proposed, and it is worth investigating polynomial-time quantum algorithms to solve LSPs.

We can also investigate Inexact-Feasible IPMs that are more adaptable with QLSAs since Feasible IPMs have better complexity than infeasible IPMs.
Another direction can be developing pure QIPMs where all calculations happen in the quantum setting.
However, current NISQ devices have some limitations that prevent having a pure QIPM. 
Such a method will circumvent the use of QTA within and enjoy the fast QLSAs.

Our computational experiments show that the proposed II-QIPM embedded in the IR scheme with the QLSA simulator of QISKIT AQUA can solve problems with hundreds of variables to a user-defined precision. 
However, a classical computer can simulate only a limited number of qubits, limiting the number of constraints. 
Although there are Feasible IPMs with better complexity than the proposed IR-II-QIPM, this paper is a significant step toward using quantum solvers in classical methods correctly and efficiently. 
In addition, the IR-II-QIPM leverage the wide neighborhood of the central path and starts with an infeasible interior solution compared to F-IPMs.
We also demonstrated for the first time that LO problems could practically be solved using quantum solvers.

\begin{acknowledgements}
This work is supported by Defense Advanced Research Projects Agency as part of the project W911NF2010022: {\em The Quantum Computing Revolution and Optimization: Challenges and Opportunities}.
\end{acknowledgements}

\printbibliography

\end{document}